\documentclass[sigconf, nonacm]{acmart}
\usepackage{enumitem}

\newcommand\vldbdoi{XX.XX/XXX.XX}
\newcommand\vldbpages{XXX-XXX}
\newcommand\vldbvolume{14}
\newcommand\vldbissue{1}
\newcommand\vldbyear{2020}
\newcommand\vldbauthors{\authors}
\newcommand\vldbtitle{\shorttitle} 
\newcommand\vldbavailabilityurl{https://github.com/ThinhOn/fair-retrieval}
\newcommand\vldbpagestyle{plain} 

\newcommand{\MAtFair}{{\bf MAFair-KNN} }
\newcommand{\MAtFairOne}{{\bf 1-Fair-KNN} }
\newcommand{\MAtFairTwo}{{\bf 2-Fair-KNN} }
\newcommand{\MAtFairThree}{{\bf 3+-Fair-KNN} }

\newcommand{\AlgNN}{{\em Alg-Near-Neighbor} }

\newcommand{\AlgNNP}{{\em Alg-Near-$\pi$} }

\newcommand{\AlgPP}{{\em Alg-PostP-$\pi$} }

\newcommand{\AlgPPOne}{{\em Post-Alg-1} }
\newcommand{\AlgPPTwo}{{\em Post-Alg-2} }
\newcommand{\AlgPPThree}{{\em Post-Alg-3+} }

\newcommand{\AlgFairOne}{{\em Alg-1-Fair} }
\newcommand{\AlgFairTwo}{{\em Alg-2-Fair} }
\newcommand{\AlgFairThree}{{\em Alg-3+-Fair} }

\usepackage{algorithm}
\usepackage[noend]{algpseudocode}
\usepackage{optidef}
\usepackage{adjustbox}
\usepackage{threeparttable}
\usepackage{subcaption}

\usepackage{mathtools}

\newcommand{\calD}{\mathcal{D}}

\newcommand{\ba}{\mathbf{a}}

\newcommand{\bx}{\mathbf{x}}

\newcommand{\bbR}{\mathbb{R}}

\newcommand{\revise}[1]{\textcolor{blue}{\textbf{#1}}}
\newcommand{\norm}[1]{\left \| #1 \right \|}
\usepackage{bbm}

\begin{document}
\title{ Multi-Attribute Group Fairness in $k$-NN Queries on Vector Databases}

\author{Thinh On}
\affiliation{%
  \institution{NJIT}
  \streetaddress{P.O. Box 1212}
  \city{Newark, NJ}
  \state{USA}
}
\email{to58@njit.edu}

\author{Baruch Schieber}
\affiliation{%
  \institution{NJIT}
  \city{Newark, NJ}
  \country{USA}
}
\email{sbar@njit.edu}

\author{Senjuti Basu Roy}
\affiliation{%
  \institution{NJIT}
  \city{Newark, NJ}
  \country{USA}
}
\email{senjutib@njit.edu}

\begin{abstract}
We initiate the study of multi-attribute group fairness in $k$-nearest neighbor ($k$-NN) search over vector databases. Unlike prior work that optimizes efficiency or query filtering, fairness imposes count constraints to ensure proportional representation across groups defined by protected attributes. When fairness spans multiple attributes, these constraints must be satisfied simultaneously, making the problem computationally hard. To address this, we propose a computational framework that produces high‐quality approximate nearest neighbors with good trade-offs between search time, memory/indexing cost, and recall. We adapt locality-sensitive hashing (LSH) to accelerate candidate generation and build a lightweight index over the Cartesian product of protected attribute values. Our framework retrieves candidates satisfying joint count constraints and then applies a post-processing stage to construct fair $k$-NN results across all attributes. For 2 attributes, we present an exact polynomial-time flow-based algorithm; for 3 or more, we formulate ILP-based exact solutions with higher computational cost. We provide theoretical guarantees, identify efficiency–fairness trade-offs, and empirically show that existing vector search methods cannot be directly adapted for fairness. Experimental evaluations demonstrate the generality of the proposed framework and scalability.

\end{abstract}

\maketitle
\pagestyle{\vldbpagestyle}
\begingroup\small\noindent\raggedright\textbf{PVLDB Reference Format:}\\
\vldbauthors. \vldbtitle. PVLDB, \vldbvolume(\vldbissue): \vldbpages, \vldbyear.\\
\href{https://doi.org/\vldbdoi}{doi:\vldbdoi}
\endgroup
\begingroup
\renewcommand\thefootnote{}\footnote{\noindent
This work is licensed under the Creative Commons BY-NC-ND 4.0 International License. Visit \url{https://creativecommons.org/licenses/by-nc-nd/4.0/} to view a copy of this license. For any use beyond those covered by this license, obtain permission by emailing \href{mailto:info@vldb.org}{info@vldb.org}. Copyright is held by the owner/author(s). Publication rights licensed to the VLDB Endowment. \\
\raggedright Proceedings of the VLDB Endowment, Vol. \vldbvolume, No. \vldbissue\ %
ISSN 2150-8097. \\
\href{https://doi.org/\vldbdoi}{doi:\vldbdoi} \\
}\addtocounter{footnote}{-1}\endgroup

\ifdefempty{\vldbavailabilityurl}{}{
\vspace{.3cm}
\begingroup\small\noindent\raggedright\textbf{PVLDB Artifact Availability:}\\
The source code, data, and/or other artifacts have been made available at \url{\vldbavailabilityurl}.
\endgroup
}
\section{Introduction}
\label{sec:intro}
Vector search has emerged as a cornerstone of modern data management, powering applications such as search engines, recommender systems, and multimodal AI pipelines~\cite{pan2024survey, han2023comprehensive, jing2024large}. State-of-the-art approximate nearest neighbor (ANN) methods—including FAISS~\cite{douze2024faiss}, HNSW~\cite{malkov2018efficient}, ScaNN~\cite{avq_2020}, and DiskANN~\cite{simhadridiskann}—deliver efficient search in high-dimensional spaces. Existing work has also studied how to efficiently retrieve top results $k$ that satisfy filtering conditions over query predicates~\cite{li2025sieve, patel2024acorn,acorn1,acorn2, acorn3,vldbpanel}. A comprehensive survey can be found in ~\cite{huang2025survey}. 

This work initiates the study of \emph{multi-attribute group fairness} in $k$-nearest neighbor ($k$-NN) search over vector databases. Prior work on vector search has largely focused on efficiency~\cite{huang2015query} and predicate-based filtering~\cite{patel2024acorn,li2025sieve,acorn1,acorn2,acorn3}. In contrast, fairness-aware retrieval requires enforcing \emph{count constraints} that guarantee proportional representation of groups defined by protected attributes. When fairness constraints span multiple attributes (e.g., gender, race, and age), they must be satisfied independently for each attribute, substantially increasing the problem complexity. Despite its growing importance, the problem of \emph{satisfying multi-attribute group fairness constraints} in $k$-NN retrieval remains largely unexplored. Yet, many high-stakes applications—including hiring, university admissions, media search, and personalized recommendation—require top-$k$ results that simultaneously preserve relevance while respecting demographic balance and diversity mandates.

\vspace{-0.1in}
\begin{example}[Fair $k$-NN Query on Bus Driver Images]
\label{example:bus-driver-attributes}
\smallskip \noindent {\bf Database.} Consider a vector database that contains images of public school bus drivers. Each image is embedded into a high-dimensional vector space using a deep neural network with $3$ protected attributes. \\
\textbf{Gender:} Male, Female, Non-binary. \\
 \textbf{Race/Ethnicity:} White, Black, Hispanic, Asian, Indigenous, Mixed. \\
\textbf{Age Group:} $<$30, 30--50, $>$50. \\
 
\noindent {\bf Query.} Suppose a parent submits a query image $\mathcal{Q}$ of a bus driver they believe drives their child’s bus. The system must return the $k=8$ most similar images.  In standard $k$-NN retrieval, the $8$ closest neighbors are selected purely by similarity. However, district policy requires \emph{fairness constraints} ensuring proportional representation across three protected attributes: gender, race, and age group. For example:  
\textbf{Gender:} 3 male, 3 female, 2 non-binary.  
\textbf{Race/Ethnicity:} 4 White, 2 Black, 2 Hispanic.  
\textbf{Age group:} 2 under $30$, 3 between $30$--$50$, 3 over $50$.  

Without fairness-aware selection, the retrieved results may be dominated by majority groups (e.g., all older white male drivers), thereby failing to reflect the actual diversity in the database. 
By enforcing fairness-aware $k$-NN, the result set remains both relevant (close to the query) and socio-demographically proportionate.
\end{example}

We initiate the study of \MAtFair as follows: given a query vector $\mathcal{Q}$ and fairness requirements across multiple protected attributes, return the set of $k$ nearest neighbors satisfying those requirements while minimizing the total distance to $\mathcal{Q}$. Practical applications of this problem include search engines aiming for demographically balanced results, e-commerce systems enforcing brand or category diversity, and content recommendation platforms ensuring proportional representation across genres or sources. We investigate three variants of the problem: {\bf 1-Fair-KNN}, which considers a single protected attribute; {\bf 2-Fair-KNN}, which extends to two protected attributes; and {\bf 3+-Fair-KNN}, which generalizes to three or more protected attributes. We prove that the satisfiability problem of \MAtFairThree itself is (strongly) NP-hard through a reduction from the 3-dimensional matching (3DM) problem~\cite{islam2022satisfying}.

Our proposed algorithmic framework operates in two main stages.  \textbf{Preprocessing:} We preprocess the vector database and design two complementary indexing structures: one optimized for efficiently enforcing fairness constraints, and the other for accelerating nearest neighbor search. \textbf{Query processing:} Given a query vector $\mathcal{Q}$, the algorithm first identifies all relevant partitions and efficiently retrieves {\em near neighbor candidates} from each. A subsequent {\em post-processing} stage enforces group fairness constraints across all specified attribute values simultaneously, ensuring that the final result set satisfies fairness requirements whenever a feasible solution exists within the data.

For ensuring efficient fairness check across any possible combination of protected attributes during query time, we partition the vector database based on the Cartesian product of the protected attribute values and index them.  Using the running example, this will partition the database into $54 = 3 \times 6 \times 3$ partitions. 

For each partition, our next goal is to enable efficient retrieval of near neighbors. Broadly, there are two families of techniques for this task: {\em graph-based indexes~\cite{malkov2018efficient,simhadridiskann,li2025sieve,patel2024acorn,acorn1,acorn2,acorn3} and locality-sensitive hashing (LSH)~\cite{huang2015query,indyk1998approximate,andoni2015practical,pham2022,charikar2002similarity,panigrahy2005entropy}. However, graph-based methods are inherently heuristic, and their performance can degrade significantly when the underlying data is not sufficiently “clusterable”~\cite{li2025sieve}. For this reason, we adopt LSH~\cite{gan2012locality,huang2015query} to support efficient approximate nearest-neighbor (NN) search within each partition. LSH is particularly suitable for high-dimensional settings and offers theoretical guarantees on recall: with enough hash tables, the probability of retrieving the true nearest neighbors can be made arbitrarily high~\cite{gionis1999similarity}.}

Leveraging these properties, we develop \AlgNN, which uses LSH to efficiently retrieve near neighbors from each partition. This structure prunes irrelevant candidates early, eliminating the inefficiencies of naïve post-filtering and ensuring that only candidates with the required attribute values are examined.

We then introduce a post-processing framework, \AlgPP, which selects the final set of $k$ outputs that jointly minimize the total distance to the query while exactly satisfying group-fairness constraints across all protected attributes, whenever a feasible solution exists. {\em Importantly, \AlgPP is agnostic to the underlying retrieval method—whether graph-based or LSH-based—and yields optimal results in either case, as confirmed by our analytical and experimental evaluation.}

Given the aforementioned framework, we present \AlgFairOne along with a polynomial time postprocessing algorithm \AlgPPOne that is designed to solve \MAtFairOne. For two-attribute queries (\MAtFairTwo), we present Algorithm \AlgFairTwo where we model the postprocessing problem \AlgPPTwo as min-cost flow in bipartite graphs, and present a polynomial time solution. For three or more attributes (\MAtFairThree), we design postprocessing algorithm \AlgPPThree as an Integer Linear Programming (ILP)–based method that guarantees exact results on fairness constraints, albeit at higher computational cost. We present analytical studies of the designed solutions.



Extensive experiments on five large-scale vector databases demonstrate that our solutions achieve quality consistent with our theoretical analyses and demonstrate  more than $3x$ to $4x$ speedup than the designed baselines. We also adapt state-of-the-art graph based $k$-NN solutions, such as, {\tt SIEVE}~\cite{li2025sieve} or Filter-DiskANN~\cite{gollapudi2023filtered} to fairness-aware setting, as well as implement fairness baselines. Our experimental evaluation demonstrates that  our proposed indexing and postprocessing algorithms to ensure fairness are generic and can be instantiated over both graph-based and LSH-based indexing paradigms. However, fairness-aware retrieval fundamentally requires dedicated indexing and candidate selection mechanisms, and existing vector search methods are not suitable there as is.

\noindent \textbf{Summary of Contributions.} The key contributions of this work are as follows:
\begin{itemize}[leftmargin=1.5em]
\item \textbf{Problem Formulation \& Analysis:}
We define the \emph{Multi-Attribute Fair $k$-NN Retrieval} (\MAtFair) problem: retrieving $k$ nearest neighbors that minimize the total distance to a query while satisfying group-fairness constraints across multiple protected attributes. We design efficient polynomial algorithms for \MAtFairOne and  \MAtFairOne and \MAtFairTwo, and prove that even finding a feasible solution of \MAtFairThree  (3 or more attributes) is \emph{strongly NP-hard} via a reduction from 3DM.

\item \textbf{Indexing Framework:}
We propose a two-level indexing framework: (i) partitioning the vector database by the Cartesian product of protected attribute values and indexing these partitions using bitmaps for fast fairness-aware filtering, and (ii) applying LSH within each partition for efficient near-neighbor retrieval.

\item \textbf{Query Processing Algorithms:}
At query time, the framework first retrieves near neighbors within each partition and then enforces fairness constraints. We design \AlgNN, an LSH-based retrieval algorithm with theoretical recall guarantees, and three fairness-enforcement algorithms—\AlgFairOne, \AlgFairTwo (min-cost flow), and \AlgFairThree (ILP)—that compute an exactly fair set of $k$ results from the retrieved candidates.

\item \textbf{Experimental Results:}
Experiments on five large-scale vector databases demonstrate that multi-attribute group fairness cannot be reliably achieved by directly adapting existing vector search based methods, such as, {\tt SIEVE}~\cite{li2025sieve} or Filter-DiskANN~\cite{gollapudi2023filtered} and that fairness-aware retrieval requires dedicated indexing and candidate selection mechanisms.
\end{itemize}

\section{Data Model \& Problems}

\subsection{Data Model}
\label{sec:data-model}
\smallskip
\noindent {\bf Protected Attributes.}
A fixed set of the attributes, denoted by $\mathcal{A} = \{A_1, \ldots, A_m\}$ are designated as protected attributes.
These are variables that correspond to legally or ethically sensitive characteristics--such as gender, age, ethnicity, disability, or socioeconomic status--that must not unjustly influence algorithmic outcomes~\cite{dwork2012fairness, barocas2023fairness}.
Each protected attribute $A_j$ has a finite domain of possible values $V_j$. Every record $(x_i, a_i)$ belongs to an intersectional group identified by the Cartesian product $V_1 \times V_2 \times \cdots \times V_m$. These protected attributes are treated as immutable and stored as metadata alongside embeddings in the database schema.

{\bf Vector databases.}
Let $\calD = \{(x_i, a_i)\}_{i=1}^n$ denote a collection of $n$ data points.
Each $x_i \in \mathbb{R}^d$ is a $d$-dimensional embedding vector in a continuous feature space, and $a_i = (a_{i,1}, \ldots, a_{i,m})$ encodes the values of $m$ protected attributes associated with $x_i$.

To ground the discussion, consider a running example where $\mathcal{D}$ is a vector database of bus driver images.
Each image is embedded into a high-dimensional feature space using a vision model and is annotated with protected attributes such as Gender (Male, Female, Non-binary), Race/Ethnicity (White, Black, Hispanic, Asian, Indigenous, Mixed), and Age Group (<30, 30–50, >50).
Table~\ref{tab:bus-driver-db} illustrates a sample of ten records from this database.
Each embedding vector $x_i$ corresponds to a driver’s image, while the tuple $a_i$ captures the demographic and employment attributes that may be relevant to fairness constraints during retrieval.


\smallskip
\noindent {\bf Distance.}
Retrieval operates on a distance function $\norm{\cdot, \cdot}: \bbR^d \times \bbR^d \rightarrow \bbR_{\ge 0}$ to measure the proximity of a data point to the query.
In the context of top-$k$ retrieval, the task reduces to a $k$-NN search, where the goal is to identify a subset of $k$ data points whose total distance to the query is minimized.
Eqn.~\ref{eq:distance-func} shows commonly used distance functions between a query vector $q$ and data point $x$.
\begin{equation}
\label{eq:distance-func}
\norm{x, q} =
\begin{cases}
\sqrt{\sum_{i=1}^d (x_i - q_i)^2} & \text{(Euclidean or $\ell_2$ distance)}\\[4pt]
\sum_{i=1}^d|x_i - q_i| & \text{(Manhattan or $\ell_1$ distance)}\\[4pt]
1 - \cos(x, q) & \text{(Cosine-based distance)}\\[4pt]
\left(\sum_{i=1}^d (x_i - q_i)^p\right)^{1/p} & \text{(Minkowski or $\ell_p$ distance)}
\end{cases}
\end{equation}

\noindent
{\bf $k$-NN Queries.}
Given a query vector $q \in \mathbb{R}^d$, a $k$-nearest neighbor ($k$-NN) query aims to identify a subset of $k$ data points in the dataset whose distances are closest to $q$ according to a specified distance function.
Formally, the query returns the set of $k$ points $\mathcal{N}_k(q) \subseteq \mathcal{D}$ such that for every $x_i \in \mathcal{N}_k(q)$ and $x_j \notin \mathcal{N}_k(q)$, the inequality $\norm{x_i, q} \le \norm{x_j, q}$ holds.

\smallskip

\begin{table}[t]
\centering
\begin{tabular}{c|c|c|c}
\hline
\textbf{ID} & \textbf{Gender} & \textbf{Race/Ethnicity} & \textbf{Age Group} \\ 
\hline
1 & Male & White & 30--50 \\
2 & Female & Black & $<$30 \\
3 & Non-binary & Hispanic & $>$50 \\
4 & Female & Asian & 30--50 \\
5 & Male & Hispanic & $<$30 \\
6 & Female & White & $>$50 \\
7 & Male & Black & 30--50 \\
8 & Non-binary & White & $<$30 \\
9 & Female & Indigenous & $>$50 \\
10 & Male & Mixed & 30--50 \\
\hline
\end{tabular}
\caption{\small Example vector database of bus driver images with protected attributes.}
\label{tab:bus-driver-db}
\end{table}

\subsection{Problem Definition}
{\bf MAFair-KNN.} Fairness constraints can be defined over any non-empty subset of the power set of $\mathcal{A}$, allowing users to impose requirements on one or more protected attributes.

Given a query $\mathcal{Q} = (q, \hat{\beta})$, where $q$ is the query vector and $\hat{\beta}$ encodes the desired group-level count constraints, the goal is to retrieve a subset $S \subseteq \mathcal{D}$ of size $k$ that minimizes total distance to $q$ while simultaneously satisfying fairness requirements:
\begin{equation}
\begin{aligned}
    \min_{S \subseteq \calD, |S| = k} \quad & \sum_{(\bx, \ba) \in S} \norm{x, q}
\end{aligned}
\end{equation}
subject to
\begin{equation}
\label{eq:fairness-constraints}
    \forall A_j \in \mathcal{A}', \forall v \in V_j: \quad \big|\{ (x, a) \in S \mid a_j = v \}\big| = \hat{\beta}_{j,v}
\end{equation}
where $\hat{\beta}_{j,v}$ denotes the required number of samples in $S$ that must have attribute value $v$ for protected attribute $A_j$, and $\sum_{v \in V_j} \hat{\beta}_{j,v} = k$.
This formulation generalizes traditional $k$-NN retrieval to a fairness-aware setting, ensuring that the returned neighbors not only minimize geometric distance but also satisfy attribute-level count constraints across multiple protected attributes.

\begin{itemize}[leftmargin=1.5em]
    \item {\MAtFairOne ($m=1$):} 
    When fairness is imposed on a single protected attribute $A$, the constraints apply only over its domain $V$, ensuring proportional representation across its groups. 
 \item {\MAtFairTwo ($m=2$):} 
    When two protected attributes $(A_i, A_j)$ are considered, the constraints apply to the domains of $V_i$ and $V_j$ independently. 
\item {\MAtFairThree($m \ge 3$):} 
    For $3$ or more protected attributes, the query answer needs to satisfy the fairness constraints on each of the attributes independently.
  \end{itemize}
Each variant exhibits distinct computational complexity and therefore requires a different algorithmic treatment.

\noindent Returning to the bus driver database in Table~\ref{tab:bus-driver-db}, suppose a parent submits a query image $q$ representing a driver. The task is to return $k = 5$ most similar images while enforcing fairness constraints over Gender, Race, and Age Group. The required counts are specified as:
\begin{itemize}[leftmargin=1.5em]
    \item Gender: $\hat{\beta}_{\text{Male}} = 1, \hat{\beta}_{\text{Female}} = 1, \hat{\beta}_{\text{Non-binary}} = 3$
    
    \item Race: $\hat{\beta}_{\text{White}} = 2, \hat{\beta}_{\text{Black}} = 1, \hat{\beta}_{\text{Hispanic}} = 2$
    
    \item Age: $\hat{\beta}_{\text{<30}} = 2, \hat{\beta}_{\text{30--50}} = 2, \hat{\beta}_{\text{>50}} = 1$
\end{itemize}
A feasible result of 5 data points must contain exactly 1 male, 1 female, and 3 non-binary images; 2 White, 1 Black, and 2 Hispanic individuals; and a 2--2--1 distribution across age groups.

Following~\cite{islam2022satisfying} we prove the following theorem. 
\begin{theorem}
{\bf 3+-Fair-KNN}: The problem of deciding the feasibility of a general instance of the 3-attribute case (and thus any $m\geq 3$ attributes as well) is strongly NP-hard.
\end{theorem}
\begin{proof}
Recall the 3-Dimensional Matching (3DM) decision problem: given three pairwise-disjoint sets $X,Y,Z$ with $|X|=|Y|=|Z|=k$ and a set $T\subseteq X\times Y\times Z$ of triples, does there exist a perfect matching $M\subseteq T$ (of size $k$) such that no two triples in $M$ share an element? 3DM is NP-complete and contains no numerical parameters, implying strong NP-completeness.

\smallskip
\emph{Reduction construction.}
From an instance $(X,Y,Z,T)$ of 3DM, build a \emph{feasibility} instance of {\bf 3+-Fair-KNN} as follows.

\begin{itemize}[leftmargin=1.5em]
\item \textbf{Attributes.} Create three protected attributes $A_1,A_2,A_3$ with domains $V_1\!=\!X$, $V_2\!=\!Y$, $V_3\!=\!Z$.
\item \textbf{Records.} For every triple $(x,y,z)\in T$, create one database record $(\mathbf{x}_{(x,y,z)}, a)$ whose protected-attribute tuple is $a=(x,y,z)$. (The embedding $\mathbf{x}_{(x,y,z)}$ and distances are arbitrary; feasibility ignores costs.)
\item \textbf{Fairness counts and $k$.} For each $x\in X$, set $\hat\beta_{1,x}=1$; for each $y\in Y$, set $\hat\beta_{2,y}=1$; for each $z\in Z$, set $\hat\beta_{3,z}=1$. This exactly matches the count-constraint template in Eq.~\ref{eq:fairness-constraints}.
\end{itemize}

This reduction is clearly polynomial in $|X|+|Y|+|Z|+|T|$.

\smallskip
\emph{Correctness.}
We show the 3DM instance is a YES-instance iff the constructed {\bf 3+-Fair-KNN} instance is feasible.

($\Rightarrow$) Suppose there exists a perfect 3D matching $M\subseteq T$ of size $k$. Select $S=\{\mathbf{x}_{(x,y,z)}:(x,y,z)\in M\}$. Because elements in $M$ do not share elements, each $x\in X$ appears in exactly one chosen triple, and likewise for $Y$ and $Z$. Therefore, the per-attribute counts satisfy $\hat\beta_{1,x}=\hat\beta_{2,y}=\hat\beta_{3,z}=1$ for all elements, and $|S|=k$. Hence $S$ is a feasible solution. Note that distances are not considered in fairness feasibility check.

($\Leftarrow$) Conversely, suppose there exists a feasible set $S$ of size $k$ satisfying the fairness counts. Each selected record in $S$ corresponds to some triple $(x,y,z)\in T$ and contributes one unit of count to exactly one $x$, one $y$, and one $z$. Since each count target is $1$ and $|S|=k$, no domain value (in $X$, $Y$, or $Z$) can appear twice among the selected records. Therefore, the set of corresponding triples is a family of $k$ mutually disjoint triples — a valid perfect 3D matching.

\smallskip
\emph{Strong NP-hardness.}
Since 3DM is known to be strongly NP-hard due to its combinatorial structure (without numerical parameters), and the above reduction preserves this property, the resulting feasibility problem for {\bf 3+-Fair-KNN} is strongly NP-hard.

\end{proof}

\section{Proposed Framework}
We propose a computational framework that unifies approximate similarity search with fairness-aware optimization. It combines attribute-partitioned indexing, LSH-based indexing, and post-processing algorithms to guarantee fairness across multiple protected attributes. At a high level, the framework operates in two main stages.  

\noindent\textbf{Stage 1: Preprocessing and Indexing.}
The dataset is partitioned by the Cartesian product of protected attribute values in $\mathcal{A}$. Each partition is indexed using locality-sensitive hashing (LSH)~\cite{indyk1998approximate,datar2004locality,huang2015query,jafari2021survey}, enabling efficient approximate nearest-neighbor search while enforcing protected-attribute constraints. Preprocessing is discussed in Section~\ref{sec:preprocess}.

\noindent\textbf{Stage 2: Query Processing.}
Query answering consists of 2 steps:
\begin{enumerate}[leftmargin=1.5em]
    \item \emph{Candidate Retrieval:} Retrieve near neighbors from partitions satisfying the fairness constraints $\hat{\beta}$ (Section~\ref{sec:qpnear}).
    \item \emph{Fair Selection:} Select exactly $k$ items that (i) minimize total distance to the query and (ii) satisfy per-attribute fairness constraints.
    The fair selection step is formulated as an optimization problem; solutions are presented in Section~\ref{sec:qppostprocess}.
\end{enumerate}

\subsection{Preprocessing}\label{sec:preprocess}
The process begins by partitioning the dataset according to the Cartesian product of all protected attribute values in $\mathcal{A}$. This results in multiple subgroups, each representing a unique combination of protected attributes (e.g., gender, race, or age categories). Such partitioning ensures that subsequent retrieval or analysis can explicitly respect and enforce fairness constraints tied to any combination of these attributes. Within each subgroup, the data is further indexed using Locality Sensitive Hashing (LSH), which provides an efficient mechanism for approximate nearest neighbor retrieval in high-dimensional vector spaces. By combining attribute-based partitioning with LSH-based indexing, the system enables both constraint-aware access to specific demographic subsets and scalable, similarity-preserving retrieval of vector representations.

\subsubsection{Attribute-Based Partitioning}
\label{sec:partition}
Recall that we have fixed set of $m$ protected attributes $\mathcal{A} = \{A_1, \ldots, A_m\}$ with corresponding domains $\{V_1, \ldots, V_m\}$.
The dataset $\mathcal{D}$ is partitioned along the Cartesian product of attribute values:
\begin{equation}
\label{eq:all-partitions}
\Pi = V_1 \times V_2 \ldots \times V_m
\end{equation}
Each partition $\pi \in \Pi$ represents an \emph{intersectional subgroup}, e.g., (Female, Asian, Age 30--50) and its subset of points is denoted as $\mathcal{D}_\pi = \{(x, a) \ | \ a = \pi\}$.
Within each partition, a separate LSH index is built using the same procedure described above.
Partitioning reduces search space by providing direct access to relevant attributes defined by query predicates and avoid searching in non-relevant partitions, thereby supporting query with multi-attribute fairness constraints while maintaining efficiency.
Although the total number of partitions may be large, each data point belongs to exactly one partition; hence, no duplication occurs and the overall storage space remains $O(n)$, ensuring reasonable space complexity.

\subsubsection{Bitmap Representation of Partitions}
\label{sec:partition-bitmap}
To efficiently represent and manipulate the partitions, we adopt a compact bitmap encoding scheme for each partition.
Each protected attribute $A_j$ with domain size $|V_j|$ is assigned a fixed bit range of $\lceil\log_2{|V_j+1|}\rceil$ bits within the overall bitmap (the added $1$ ensures absence of any value for a protected attribute). 
Each possible value $v \in V_j$ is mapped to a unique binary code within that range.
A partition $\pi = (v_1, v_2, \ldots, v_m)$ can therefore be represented by concatenating the binary codes of its constituent attribute values, forming a compact binary vector $b_\pi$.
This representation allows efficient bitwise operations for partition selection during query processing.

Running Example~\ref{example:bus-driver-attributes}: $A_1 = \text{Gender}, V_1 = \{\text{Male}, \text{Female}, \text{Non-binary}\}, \\ \lceil\log_2{|V_1+1|}\rceil = 2$. Hence, the attribute values can be encoded using $2$ bits: $\text{Male} = 01, \text{Female} = 10, \text{Non-binary} = 11, \\ \text{Absence of a gender value} = 00$. \\
 Similarly, for $A_2 = \text{Race/Ethnicity}$: \text{White} = 001, \text{Black} = 010,\\ \text{Hispanic} = 011, \text{Asian} = 100, \text{Indigenous} = 101, \text{Mixed} = 110$, \\\text{Absence of a race value} = 000$. \\
$A_3 = \text{Age Group}$: $\text{<30} = 01, \text{30--50} = 10, \text{>50} = 11, \\ \text{Absence of an age value} = 00$. \\
Partition $\pi = (\text{Female}, \text{Hispanic}, \text{30--50})$ can be represented by the bitmap $b_\pi = \boxed{10}\boxed{011}\boxed{10} = 1001110$.

\subsubsection{Locality-Sensitive Hashing}
\label{sec:lsh}
Given two vectors $x, y \in \mathbb{R}^d$, a hash function $h$ is $(R, cR, p_1, p_2)$-sensitive for a distance metric $\norm{\cdot, \cdot}$ if the following two conditions hold:
\begin{gather}
\begin{cases}
    \mathbb{P}\left[ h(x) = h(y) \right] \geq p_1, \quad \textrm{if } \norm{x, y} \leq R
    \\
    \mathbb{P}\left[ h(x) = h(y) \right] \leq p_2, \quad \textrm{if } \norm{x, y} > cR
\end{cases}
\end{gather}
where $c > 1$ is an approximation factor, $p_1$ and $p_2$ are near and far collision probabilities and $p_1 > p_2$.
Intuitively, two points $x$ and $y$ are projected to the same bucket (\emph{i.e.} they collide) with probability at least $p_1$ if the distance between them is less than $R$. Two far away points with distance greater than $cR$ can still be projected in the same bucket but with probability at most $p_2$.
LSH schemes differ by the underlying metric:
\begin{itemize}[leftmargin=1.5em]
    
    \item {\bf $\ell_2$ distance}: under $\ell_2$ distance for points in $R^d$, $p$-stable LSH~\cite{datar2004locality} proposed using the base hash
    \begin{equation}
        h(x) = \left\lfloor \frac{a \cdot x + b}{w} \right\rfloor
    \end{equation}
    where $a \sim \mathcal{N}(0, I_d)$ is a random vector and $b$ is uniformly drawn from $[0, w)$. 

    \item {\bf Arccos (angular distance):} pick a random unit length vector $u \in \mathbb{R}^d$, the base hash function under angular metric is defined by~\cite{charikar2002similarity}:
    \begin{equation}
        h(x) = \text{sign}(u \cdot x)
    \end{equation}
    Under this distance metric, the collision probability is $\mathbb{P}[h(x) = h(y)] = 1 - \frac 1\pi
    \arccos(\frac{x \cdot y}{\norm{x} \cdot \norm{y}})$
\end{itemize}

Following standard practice in $p$-stable LSH~\cite{datar2004locality}, 
we amplify selectivity by concatenating $\mu$ base hash functions to form a compound hash function $g(x) = (h_1(x), \ldots, h_\mu(x))$.  This compound function maps each data point $x$ to a hash key $g(x)$ that 
uniquely identifies a bucket within a hash table.  

To further reduce false negatives, we construct $\ell$ independent hash tables $T_1, T_2, \ldots, T_\ell$, each associated with its own compound hash function 
$g_j(x)$.  
Consequently, each data point $x_i$ is inserted into $\ell$ different buckets, 
one per table, according to the following standard insertion rule:
\begin{equation*}
    \text{For table } j = 1, \ldots, \ell: \quad 
    \text{insert } (g_j(x_i), x_i) \text{ into bucket } T_j[g_j(x_i)].
\end{equation*}
The preprocessing procedure is illustrated in Figure~\ref{fig:preprocessing}.
We follow this standard LSH setup to ensure a good balance between retrieval efficiency and accuracy.


\subsubsection{Indexing Cost Analysis}
Let $\Pi = \{\pi\}$ be the set of partitions and $n_\pi = |D_\pi|$ be the number of points in partition $\pi$.
Each partition $\pi$ builds $\ell$ hash tables, and each table uses a compound key obtained by concatenating $\mu$ hash functions (see Section~\ref{sec:lsh} and Figure~\ref{fig:preprocessing}).
Each base hash costs $O(d)$ time; computing one compound key thus costs $O(\mu d)$.
Since each data point must be inserted in to all $\ell$ tables in a partition, the per-point indexing cost is $O(\ell \mu d)$.
Therefore, the total indexing time over all partitions is $O\left( n \ell \mu d \right)$.

In terms of storage, each partition $\pi$ costs $O(n_\pi \ell\mu)$, hence the total space complexity across all partitions is $O\left(n \ell \mu\right)$.

\begin{figure}[t]
    \centering
    \includegraphics[width=1.0\linewidth, trim=0.7cm 18.5cm 0.65cm 1.25cm, clip]{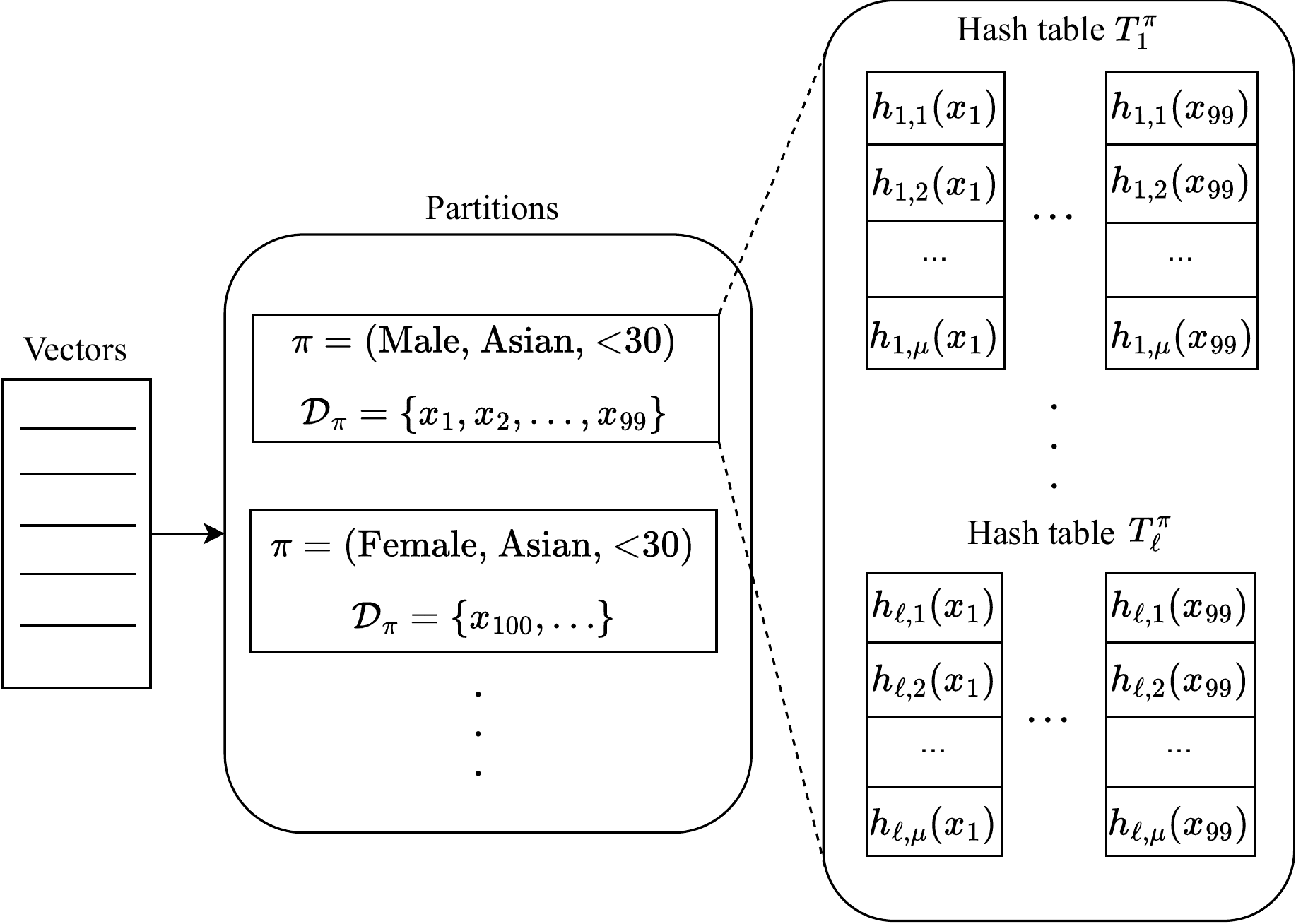}
    \caption{\small Preprocessing vectors: partitioning by Cartesian product of attribute values and locality-sensitive hashing. For each partition $\pi$, we employ $\ell$ concatenated hash functions, each of length $\mu$, to form $\ell$ hash tables.}
    \label{fig:preprocessing}
\end{figure}



\section{Query Processing Algorithms}
\label{sec:query-processing}
Query processing proceeds in two steps.  
In the first step (Section~\ref{sec:qpnear}), given a query $\mathcal{Q} = (q, \hat{\beta})$ with fairness constraints $\hat{\beta}$, Algorithm \AlgNN identifies partitions $\pi \in \Pi$ whose protected-attribute combinations are required by $\hat{\beta}$ and retrieves the sufficient number of {\em near neighbor candidates from each}. 
In the second step (Section~\ref{sec:qppostprocess}), all candidates retrieved from the relevant partitions are passed to a  \emph{post-processing algorithms (\AlgPPOne, \AlgPPTwo, \AlgPPThree)}, to select $k$ final results that (i) minimize the total  distance to the query vector and (ii) satisfy the fairness constraints  across all protected attributes as specified in $\hat{\beta}$. A key novelty of the postprocessing algorithms is that they form an optimal, stand-alone framework: given any pool of retrieved near-neighbor candidates from all relevant partitions, the algorithms produce solutions that are provably optimal with respect to that pool. In Section~\ref{sec:qpanalysis}, we present qualitative analyses of the designed algorithms.

\subsection{Retrieving Near Candidates}\label{sec:qpnear}
We present Algorithm \AlgNN for retrieving {\em near} candidates of $q$ from  a $\pi \in \Pi$, where $\pi$ represents a partition containing $\hat{\beta}$.  \AlgNN has a subroutine \AlgNNP  that retrieves a candidate set $C_{\pi}$ with number of candidates $k_{\pi}$ from partition $\pi$ ($k_{\pi} \leq k$). After running \AlgNN for all constraints in $\hat{\beta}$, a total candidate set $\mathcal{C}= \bigcup_{\forall \pi} C_{\pi}$ is produced.

Using Example~\ref{example:bus-driver-attributes}, {\em male-white-under 30} represents one such partitions, from where $C_{\pi}$ near neighbors are retrieved by \AlgNNP.

\subsubsection{Identifying Relevant Partitions via Bitmap Matching}
Recall that for a query $\mathcal{Q} = (q, \hat{\beta})$, the fairness constraints $\hat{\beta}_{j,v} \in \hat{\beta}$ specifies count constraints  of a protected attribute value $v$ of a protected attribute $A_j \in \mathcal{A}$. 
\AlgNN computes the Cartesian product across all protected attributes 
$A_j$ that appear in the fairness specification $\hat{\beta}$, producing 
a collection of partitions $\Pi_{\hat{\beta}}$. Each partition $\pi_q \in \Pi_{\hat{\beta}}$ corresponds to a unique combination of protected attribute values. For example, if $\hat{\beta}$ specifies count constraints over both 
\emph{Gender} and \emph{Race}, then $\Pi_{\hat{\beta}}$ consists of one partition  for each pair $(v_{\text{Gender}}, v_{\text{Race}})$ obtained from the Cartesian product of their respective value sets. 
\AlgNN then applies the same bitmap-based representation to every such partition. Since Age Group is unspecified, it gets the value of $00$. This results in a $\{b_{Q_1}, b_{Q_2}, \ldots\, b_{Q_{\Pi_{\hat{\beta}}}}\}$ bitmap representations of $\hat{\beta}$ for a $\mathcal{Q}$.

To determine which partitions should be probed, \AlgNN takes each preprocessed  partition bitmap $b_\pi$ and first applies an attribute mask $M_Q$ that zeros out all protected-attribute positions not constrained by the query $b_{Q_i}$. It then performs a bitwise \textsc{XOR} operation $(\oplus)$ between $b_{Q_i}$ and 
the masked partition bitmap. A partition $\pi$ is considered relevant if there exists  one $b_{Q_i}$ such that
\[
((b_\pi {\wedge M_Q}) \oplus b_{Q_i}) = 0.
\]
This condition guarantees that the partition contains all attribute 
values required by at least one fairness-constrained group. The set of 
relevant partitions is thus:
\begin{equation}
\label{eq:bitmap-relevant}
\Pi_{\hat{\beta}} = 
\{\pi \mid \exists\, b_{Q_i} \in \mathcal{B}_Q,\ 
(b_\pi \oplus b_{Q_i}) = 0\}.
\end{equation}
These partitions are subsequently probed in the retrieval stage, ensuring 
that group-level fairness constraints are enforced without sacrificing 
query efficiency.
\begin{example}
\label{example:query-bitmap}
    Consider a query with fairness constraints:
    \begin{equation}
        \nonumber
        \hat{\beta} = 
        \left\{
        \begin{array}{c}
        A_1 =\text{Gender: } \hat{\beta}_{\text{Male}} = 2, \hat{\beta}_{\text{Female}} = 3 \\[4pt]
        A_2 = \text{Race: } \hat{\beta}_{\text{Hispanic}} = 4, \hat{\beta}_{\text{White}} = 1\\
        \end{array}
        \right\}
    \end{equation}
    The corresponding set of query bitmaps is:
    \begin{equation}
        \nonumber
        \mathcal{B}_Q = 
        \left\{
        \begin{array}{c}
        b_{(\text{Male, Hispanic, \_\_})}\\
        b_{(\text{Female, Hispanic, \_\_})}\\
        b_{(\text{Male, White, \_\_})}\\
        b_{(\text{Female, White, \_\_})}
        \end{array}
        \right\}
        =
        \left\{
        \begin{array}{c}
        00010\_ \, \_\\
        01010\_ \, \_\\
        00000\_ \, \_\\
        01000\_ \, \_
        \end{array}
        \right\}
    \end{equation}
    Since $\hat{\beta}$ does not specify any fairness constraint on the protected attribute $A_3 = \text{Age Group}$, the last two bits in each query bitmap remain unassigned (shown as underscores), allowing for any age group value to be considered relevant.
\end{example}

\subsubsection{Subroutine \AlgNNP}
Given a partition $\pi$, Subroutine \AlgNNP retrieves candidate set $C_{\pi}$ with  $k_{\pi}$ number of candidates from partition $\pi$. 

Recall that partition $\pi$ contains $\ell$ hash tables $\{T_j^\pi\}_{j=1}^\ell$, constructed using compound LSH functions $g_j(\cdot)$ as defined previously. Given a query vector $q \in \mathbb{R}^d$, the algorithm first computes, for every table $T_j$, the compound hash key $g_j(q) = (h_{j, 1}(q), \ldots, h_{j, \mu}(q))$. These keys identify the buckets in which potential neighbors of $q$ are likely to reside.

For each partition $\pi$, the algorithm probes the corresponding buckets $T^\pi_{j}[g_j(q)]$ across all tables and aggregates the collision points to form the candidate set $C_\pi = \bigcup_{j=1}^\ell T_j^\pi[g_j(q)]$.
Each candidate in $C_\pi$ is then evaluated by computing its distance to query, and the results are ranked according to distance.
Algorithm~\ref{alg:retrieve-near-points} describes the procedure of retrieving $k_\pi$ near points from a partition $\pi$.
When a partition $\pi$ does not contain enough points to satisfy $k_\pi$, we return all available candidates from $\pi$ as the final result for that partition.

\begin{algorithm}[t]
\caption{Subroutine: \AlgNNP retrieve $k_{\pi}$ points in partition $\pi$}
\label{alg:retrieve-near-points}
\begin{algorithmic}[1]
\Require query $q\in\mathbb{R}^d$; hash tables $\{T_j^\pi\}_{j=1}^{\ell}$ with compound hash functions $g_j(\cdot)$; target $k_\pi$
\Ensure $k_\pi$ near points of $q$
\State $C_\pi \gets \varnothing$ \Comment{candidate set for partition $\pi$}
\For{$j \gets 1$ \textbf{to} $\ell$}
    \State $q_h \gets g_j(q)$ \Comment{compound hash key}
    \State $C_\pi \gets C_\pi \cup T_j^\pi[q_h]$ \Comment{gather collision points}
\EndFor
\State keep $k_\pi^*$ points in $C_\pi$ $(k^*_\pi > k_\pi)$, remove the rest
\State compute $\norm{x,q}$ for all $x \in C_\pi$, sort ascending
\State \Return first $k_\pi$ elements of $C_\pi$
\end{algorithmic}
\end{algorithm}
\subsubsection{Size of $k_{\pi}$}
\label{sec:k_pi}
Recall that $\hat{\beta}_{j, v}$ denote the required count for value $v$ of attribute $A_j$. For a relevant partition $\pi = (v_1, v_2, \ldots, v_m)$ $\in \Pi_{\hat{\beta}}$, then the retrieval quota for $\pi$ is:
\begin{equation}
\label{eq:k_pi}
    k_\pi = \min_{A_j \in \mathcal{A}} \hat{\beta}_{j, v_j} \quad 
\end{equation}

\begin{example}
We continue Example~\ref{example:query-bitmap}.
Any partition having attribute values Male and White--for instance $\pi = (\text{Male, White, <30})$--is relevant to the query since it can match $b_{(\text{Male, White, \_\_})}$.
By Eq.~\ref{eq:k_pi}, $k_\pi = \min (\hat{\beta}_{\text{Male}}, \hat{\beta}_{\text{White}}) = \min(2, 1) = 1$.
\end{example}

The proofs could be found in our Technical Report~\cite{techreport}.
\begin{theorem}[Number of Hash Tables in \AlgNNP]
\label{theorem:ell}
   Given an $(R, cR, p_1, p_2)$-sensitive hash family, let $n_\pi$ denote the number of points in partition $\pi$ and $\mu_\pi$ the concatenation length, where 
\[
\mu_\pi = \frac{\log n_\pi}{\log(1/p_2)} 
\quad \text{and} \quad 
\rho = \frac{\log(1/p_1)}{\log(1/p_2)}
\]
\cite{datar2004locality}.  
Let $K$ be the maximum number of near points to be retrieved for any query with success probability at least $1 - \frac{\delta}{2}$.  
Then the required number of hash tables for partition $\pi$, denoted $\ell_\pi$, must satisfy:
    \begin{equation}
    \label{eq:ell}
        \ell_\pi \ge n_\pi^{\rho} \log(2K/\delta)
    \end{equation}
\end{theorem}

Another failure mode arises when the retrieval set is dominated by false positives—points whose distance from the query exceeds $cR$.  
Since retrieving exactly $k_\pi$ points may not guarantee the presence of $k_\pi$ true near neighbors, it is necessary to examine more than $k_\pi$ candidates.  In this case, we bound the expected number of false positives and determine the minimum number of retrieved points required to ensure that at least $k_\pi$ near points are included.

\begin{theorem}[Expected Number of False Positives.]
\label{theorem:k-star}
    With probability at least $1 - \frac{\delta}{2}$, there are at most $\left\lceil \frac{2\ell_\pi}{\delta}\right\rceil$ false positives in the retrieval set. Hence, the minimum number of points to guarantee at least $k_\pi$ near point is given by:
    \begin{equation}
    \label{eq:k-star}
        k^*_\pi = k_\pi + \left\lceil\frac{2 \ell_\pi}{\delta}\right\rceil
    \end{equation}
\end{theorem}

\begin{theorem}[Running Time of \AlgNNP]
\label{theorem:alg-pp}
    Let partition $\pi$ have $n_\pi$ vectors, each of $d$-dimensional, and $\ell_\pi$ LSH tables using $\mu$ concatenated base hashes per table. Let $k_\pi$ be the quota for $\pi$, and let $k^*_\pi$ be defined by Eq.~\ref{eq:k-star}. Assume $\ell_\pi$ satisfies Eq.~\ref{eq:ell}, with probability at least $1 - \delta$, \AlgPP runs in $O(\ell_\pi \mu d + k^*_\pi d + k^*_\pi \log k_\pi)$.
\end{theorem}

\subsubsection{Algorithm \AlgNN}\label{corner}
Overall pseudo-code of \AlgNN is presented in Algorithm~\ref{algnn}.
Due to the approximate nature of LSH, some partitions may yield fewer than $k_\pi$ collision points. All available points from $\pi$ are included in the candidate pool $\mathcal{C}$, even when this number is strictly smaller than $k_\pi$. Note that the \MAtFair problem is solved globally across all partitions, and feasibility does not require that each partition meet its local target $k_\pi$ individually, only that the union $\mathcal{C}$ contains a feasible solution of size $k$.

\begin{algorithm}[H]
\caption{\AlgNN - Retrieve near candidates for a given $\mathcal{Q}$}
\label{algnn}
\begin{algorithmic}[1]
\Require query $Q = (q,\hat{\beta})$, 
         index structures $\{b_\Pi\}$
\Ensure $\mathcal{C}$
\State $\mathcal{C} \gets \emptyset$
\State identify partition set $\Pi_{\hat{\beta}}$
\For{each partition $\pi \in \Pi_{\hat{\beta}}$}
    \State $\mathcal{C}_{\pi} \gets$ \AlgNNP($b_\Pi, \pi_q$) 
     \State $\mathcal{C} \gets \mathcal{C} \cup \mathcal{C}_{\pi}$
\EndFor
\State \Return $\mathcal{C}$
\end{algorithmic}
\end{algorithm}

\begin{corollary}[Running Time of \AlgNN]
Let $\Pi_{\hat{\beta}}$ be the set of partitions selected for probing by the fairness constraints in $\hat{\beta}$.
With probability at least $1 - \delta$, running \AlgNN, which calls \AlgNNP for each $\pi_q \in \Pi_{\hat{\beta}}$, takes:
\begin{equation}
\label{eq:algnn-complexity}
O\left( |\Pi| \cdot |\mathcal{B}_{Q}| + \sum_{\pi_q \in \Pi_{\hat{\beta}}} \ell_\pi \mu d + k_\pi^*d + k^*_\pi \log k_\pi\right)
\end{equation}
\end{corollary}
\begin{proof}
    Apply Theorem~\ref{theorem:alg-pp} to each $\pi \in \Pi_{\hat{\beta}}$. To identify $\Pi_{\hat{\beta}}$, we need to scan all partitions and all query bitmaps once, which takes $O(|\Pi| \cdot |\mathcal{B}_{Q}|)$. To ensure a joint success probability $\ge 1 - \delta$ over all partitions, allocate failure probability $\delta / |\Pi_{\hat{\beta}}|$ per partition and apply a union bound across $|\Pi_{\hat{\beta}}|$ calls. The per-partition running time then sums linearly over $\pi$, yielding the stated total time complexity. 
\end{proof}

\subsection{Postprocessing Algorithms}\label{sec:qppostprocess}
Given the candidate set $\mathcal{C}$, we now describe the postprocessing algorithms \AlgPP used to enforce the fairness constraints when selecting the final results from $\mathcal{C}$. As mentioned in Section~\ref{sec:query-processing}, any technique to retrieve near candidates could be used to prepare inputs $\mathcal{C}$ for the postprocessing algorithms. Therefore, if $\mathcal{C}$ contains the $k_\pi$ true nearest neighbors from each partition, then the algorithms optimally solve the corresponding \MAtFair problem. In the discussion below, we assume that $\mathcal{C}$ contains $k_\pi$ points from each partition $\pi$. We later on lift that assumption and discuss what happens when that is not the case.

We first present \AlgPPOne, the postprocessing method for \MAtFairOne. 
Next, \AlgPPTwo\ provides a network-flow based postprocessing solution for \MAtFairTwo. Finally, \AlgPPThree\ gives an ILP  based postprocessing solution for \MAtFairThree.

\subsubsection{Algorithm \AlgPPOne}
When fairness constraints involve a single protected attribute $A_j$, each data point $x \in \mathcal{C}$ belongs to exactly one protected attribute value   $v$ of $A_j$ with a required count constraint $\hat{\beta}_{j, v}$. The goal is to select exactly $k = \sum_{v \in V} \hat{\beta}_{1, v}$ points minimizing the total distance to the query $q$ while meeting every group’s required count.
Because the fairness constraint is separable across groups, the problem decomposes into independent top-$\hat{\beta}_{j, v}$ retrievals per attribute value.

For each $v$, \MAtFairOne computes the distances $\norm{x_i, q}$ for all $x \in \mathcal{C}_v$ and sorts them in ascending order. The algorithm then keeps the $\hat{\beta}_{j,v}$ closest items from each $v$. The union of these selected subsets forms the final results.

{\bf Running time.} Let $n_v$ be the number of points for attribute value $v$. Distance computation over all points costs $O(\sum_v n_vd)$ and sorting costs $O(\sum_v n_v \log n_v)$. Since $\sum_v n_v = |\mathcal{C}|$, \MAtFairOne  takes $O(|\mathcal{C}|d + |\mathcal{C}| \log |\mathcal{C}|)$.

\subsubsection{Algorithm \AlgPPTwo}
Algorithm \AlgPPTwo is the postprocessing algorithm designed when fairness constraints span two protected attributes $A_i$ and $A_j$,  the problem can be modeled as a bipartite network $G = (U_i, U_j, E)$ given the candidate set $\mathcal{C}$, where:
\begin{itemize}[leftmargin=1.5em]
    \item $U_i = \{ U_{i, v} \mid v \in V_i \}$ represents nodes corresponding to the possible values of $A_i$;
    \item $U_j = \{ U_{j, t} \mid t \in V_j \}$ represents nodes corresponding to the possible values of $A_j$;
    \item for each candidate $x_p \in \mathcal{C}$ with attributes $(a_{p, i}=v, a_{p, j} = t)$, create an edge $e_p = (U_{i,v} \rightarrow U_{j,t})$ with cost $c_p = \norm{x_p, q}$.
\end{itemize}
The fairness-constrained $k$-NN problem can  be expressed as a min-cost flow problem in the bipartite network $G$ where each node in $U_i$ is a source and each node $U_j$ is a sink.
Let $E$ be the multiset of edges $e_p$ (one per candidate). 
Each edge $e_p$ is assigned capacity 1. A unit flow in $e_p$ implies 
the selection of candidate $x_p$. (It is well known that the flow is always integral.) The resulting formulation.
\vspace{-0.05in}
\begin{equation}
\label{eq:network-flow-problem}
\begin{aligned}
    \min_{f_p} &\sum_{p} f_p \, c_p \\
    \text{s.t.} \quad
        \forall v \in V_i: \sum_{p \ : \ a_{p,i}=v} f_p &= \hat{\beta}_{i,v},
        \quad \text{(fairness for } A_i) \\
        \forall t \in V_j: \sum_{p \ : \ a_{p,j}=t} f_p &= \hat{\beta}_{j,t},
        \quad \text{(fairness for } A_j) \\
        \forall p\in E: \quad f_p \in \{0,1\}
\end{aligned}
\end{equation}


Given this modeling of the problem, \MAtFairTwo then adapts a \emph{min-cost max-flow} algorithm~\cite{chen2023almost} which has polynomial time complexity of $O(|E|^{1 + o(1)})$, close-to-linear w.r.t. the number of bipartite edges $|E|$. An example bipartite graph for two-attribute case is illustrated in Figure~\ref{fig:network-flow}.

\begin{figure}[t]
    \centering
    \includegraphics[width=1.0\linewidth, trim=0.74cm 13.95cm 1.6cm 2.6cm, clip]{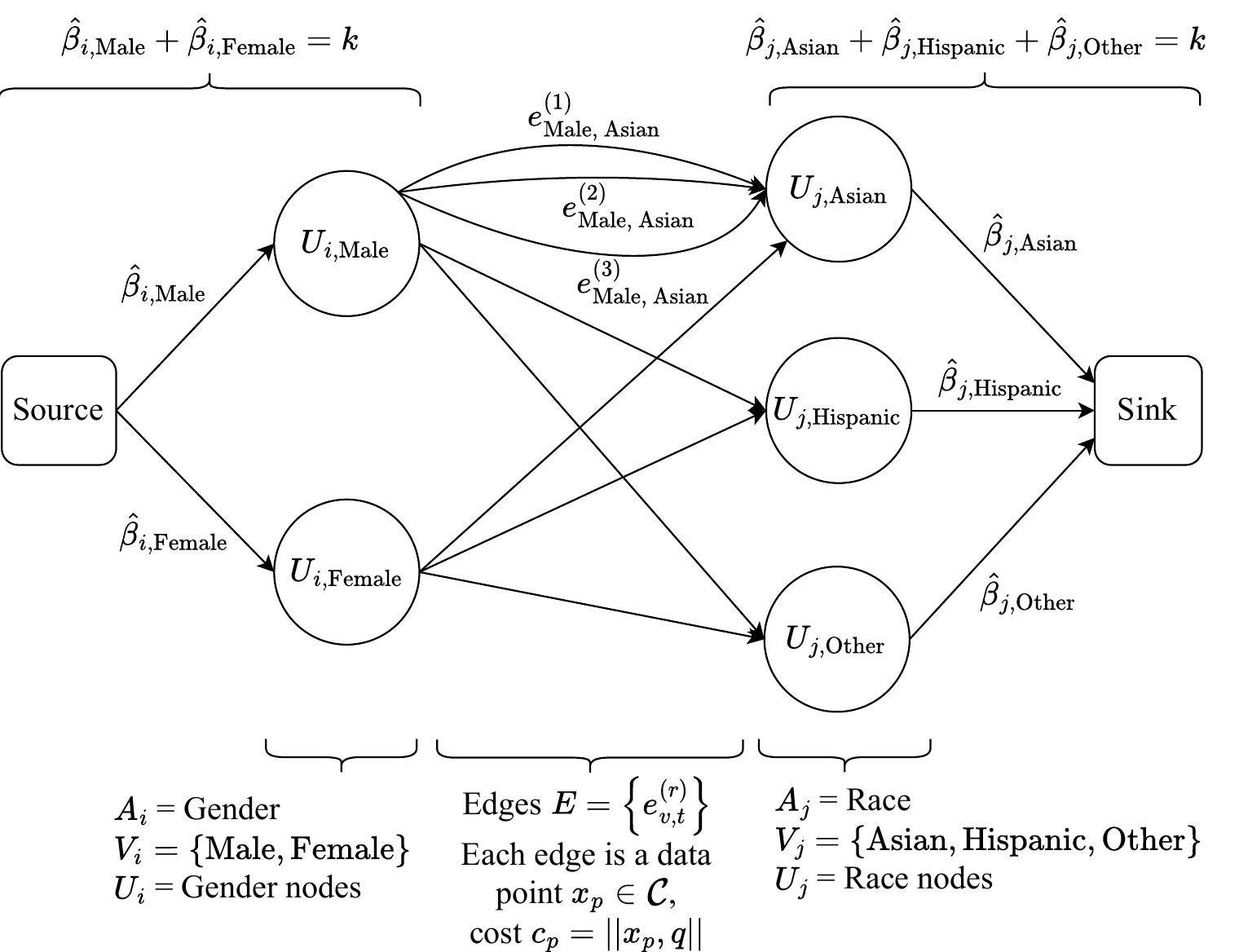}
    \caption{\small An example bipartite graph for two-attribute case. To satisfy fairness constraints, we select $k$ edges from $E$ that respect the $\hat{\beta}$ values demanded by Source and Sink. Note that it is possible to have multiple parallel edges for each pair $(v, t)$, one per data point in the candidate set $\mathcal{C}$, as shown for the pair $(v, t) = (\text{Male, Asian})$ where there are 3 candidates being Male Asian.}
    \label{fig:network-flow}
\end{figure}




\subsubsection{Algorithm \AlgPPThree}
We formulate the $k$ candidate selection problem from $\mathcal{C}$ as an Integer Linear Program (ILP) with the following goal:
Minimize the total distance between the query $q$ and the selected items, i.e., find a set of $k$ nearest neighbors that collectively minimize $\sum_i z_i \cdot \norm{x_i, q}$, while satisfying the fairness constraints. The ILP is formalized as follows:

\begin{equation}
\begin{aligned}
\text{minimize} \quad & \sum_{i=1}^{|\mathcal{C}|} z_i \cdot \norm{x_i, q} \\
\text{subject to} \quad 
& \sum_{i=1}^{|\mathcal{C}|} z_i = k, \\
& \forall j \in [m], \forall v \in V_j: \sum_{i: a_{i,j} = v} z_i = \hat{\beta}_{j,v}, \\
& z_i \in \{0,1\}, \quad \forall i \in [|\mathcal{C}|].
\end{aligned}
\end{equation}

\noindent
\textbf{Explanation.}
\begin{itemize}[leftmargin=2em]
    \item $\mathcal{C}$: the set of candidate points retrieved by \AlgNN.
    \item $x_i$: the $i$-th candidate in $\mathcal{C}$.
    \item $q$: the query vector.
    \item $\norm{x_i, q}$: the distance between $q$ and $x_i$.
    \item $z_i \in \{0,1\}$: a binary selection variable indicating whether candidate $x_i$ is selected ($z_i = 1$) or not ($z_i = 0$).
    \item $k$: the total number of items to be selected.
    \item $m$: the number of protected attributes.
    \item $V_j$: the set of possible values for the $j$-th protected attribute.
    \item $a_{i,j}$: the value of the $j$-th protected attribute for candidate $x_i$.
    \item $\hat{\beta}_{j,v}$: the desired number of items belonging to group $(j,v)$ to ensure fairness.
\end{itemize}

\noindent
\textbf{Constraints.}
\begin{enumerate}[leftmargin=2em]
    \item The first constraint $\sum_i z_i = k$ ensures exactly $k$ items are selected.
    \item The second constraint enforces fairness: for each protected attribute $j$ and each of its values $v$, the number of selected candidates from group $(j,v)$ equals $\hat{\beta}_{j,v}$.
    \item The binary constraint $z_i \in \{0,1\}$ ensures that each candidate is either selected or not.
\end{enumerate}

\subsubsection{Final Algorithms \& A Corner Case}
\AlgFairOne\ is obtained by first retrieving the near candidates using \AlgNN,  followed by postprocessing with Algorithm~\AlgPPOne. 
Similarly, \AlgFairTwo\ retrieves the near candidates using \AlgNN\ and then  applies Algorithm~\AlgPPTwo\ for postprocessing. 
Finally, \AlgFairThree\ is constructed by invoking \AlgNN\ to retrieve the near candidates and postprocessing them using Algorithm~\AlgPPThree.

Recall Section~\ref{corner}, due to the approximate nature of LSH used in \AlgNN, some partitions may produce fewer than $k_\pi$ collision points. In such cases, the postprocessing algorithms do not discard the partition $\pi$ or abort early. Instead, all available points from $\pi$ are added to the candidate pool $\mathcal{C}$, and postprocessing continues as usual.

After \AlgFairOne,\AlgFairTwo, or \AlgFairThree produces an answer set, the algorithm checks whether the fairness constraints are satisfied. If the constraints are violated, the query is reported as failed; otherwise, the query is deemed successful and the final results are returned.
\subsection{Qualitative Guarantees}\label{sec:qpanalysis}
 In Section~\ref{sec:analysisalone}, we evaluate the qualitative properties of the postprocessing algorithms as a stand alone framework. In Section~\ref{sec:analysistogether}, we evaluate the quality of the overall framework, i.e., near candidates retrieved using \AlgNN followed by the postprocessing algorithms. 

\subsubsection{Qualitative Guarantees of the Postprocessing Algorithms}\label{sec:analysisalone}
All proofs could be found in our technical report~\cite{techreport}. We prove below that the proposed postprocessing algorithms exactly satisfy the fairness constraints, provided a feasible fair solution of size $k$ exists in the candidate set $\mathcal{C}$.

\begin{theorem}
\AlgPPOne is optimal given the candidate set $\mathcal{C}$.
\end{theorem}

\begin{theorem}
\AlgPPTwo is optimal given the candidate set $\mathcal{C}$.
\end{theorem}

\begin{theorem}
\AlgPPThree is optimal given the candidate set $\mathcal{C}$.
\end{theorem}

\subsubsection{Qualitative Guarantees of the Overall Algorithmic Framework}\label{sec:analysistogether}

\begin{theorem}
If all near points are retrieved with probability at least 
$1 - \frac{\delta}{2}$ from a partition using \AlgNN and then postprocessed using the algorithms in Section~\ref{sec:qppostprocess}, then the probability that \AlgFairOne \AlgFairTwo \AlgFairThree return $C$ approximate nearest neighbors  of a query is  $1 - \frac{\Pi \delta}{2}$.
\end{theorem}

\section{Experimental Evaluations}
Our experimental evaluation has 3 primary goals - (i) despite appropriate adaptation, we demonstrate where graph based vector search techniques designed for near neighbor search fails. (ii) We evaluate the query processing qualitatively and scalability wise of our proposed algorithmic framework and compare with appropriate baselines. (iii) We report the pre-processing cost of the proposed algorithmic framework.

\subsection{Experimental Setup}
\noindent {\bf Software and hardware.} All algorithms are implemented in Python and executed on a Linux machine of $128$GB in RAM using a single CPU core AMD EPYC 7753.
For reproducibility, we provide the full implementation and scripts at: \url{https://github.com/ThinhOn/fair-retrieval}. We present a subset of results that are representative.

\noindent {\bf Datasets.}
We use four real-world datasets and one synthetic dataset (Table~\ref{tab:datasets}) that span diverse application domains and fairness-relevant characteristics. The real-world datasets include FairFace~\cite{fairface} and CelebA~\cite{celeba} for image embeddings, Audio~\cite{Audio} for speech representations, and GloVe~\cite{GloVe} for text embeddings.
We also construct a controlled synthetic dataset. Specifically, we fix a reference center in the embedding space and generate two types of points: (i) a very small  cluster whose vectors are tightly concentrated around this center, and (ii) a large set of remaining points that are located far from the center, but are themselves organized into another well-formed cluster.

\noindent {\bf Algorithms.}
We implement following types of algorithms.
\begin{enumerate}[leftmargin=2em]
   \item {\bf Nearest-neighbor retrieval (No Fairness).} These algorithms perform standard $k$-NN retrieval without fairness constraints.
   \begin{itemize}[leftmargin=0.5em]
       \item Brute-force $k$-NN: creates partitions using our proposed techniques in Section~\ref{sec:qpnear}. Given the query, it computes exact distances to all points for all relevant partitions and returns the actual top-$k$. 
       \item \textsc{SIEVE}~\cite{li2025sieve}: a state-of-the-art graph-based ANN method optimized for fast k-NN retrieval under query predicates, without fairness guarantees.
       \item \textsc{Filter-DiskANN}~\cite{gollapudi2023filtered}: extends \textsc{DiskANN}~\cite{simhadridiskann} with attribute-based filtering during graph traversal. Nodes that do not satisfy query constraints are skipped, but no guarantees are provided for satisfying multi-attribute count constraints.
       \item Our LSH based near candidate retrieval \AlgNN as discussed in Section~\ref{sec:qpnear}.
   \end{itemize}

   \item {\bf Fairness aware retrieval.} These are designed to ensure fairness.
   \begin{itemize}[leftmargin=0.5em]
       \item Single-Attribute Indexing + Retrieval ({\tt SAIR}): The process builds index per protected attribute. Each record gets indexed multiple times in that way, once per protected attribute. Given a query, it finds $k$ points per constraint,  following which it attempts to identify if the intersection of the retrieved points satisfies all the constraints in the query. If it does, the results are returned, else failed. 
       \item Joint Indexing + Retrieval ({\tt JIR})~\cite{stoyanovich2018online}: This assumes independence among protected attributes and assigns retrieval budget to each relevant partition $\pi \in b_Q$ (defined in Eq.~\ref{eq:bitmap-relevant}) based on the product of marginal proportions of the constrained attribute values encoded in $b_Q$.
       Specifically, if $b_Q$ imposes constraints on a subset of attributes $\mathcal{A}_Q \subseteq \mathcal{A}$, for each relevant partition $\pi = (a_1, \ldots, a_m)$, {\tt JIR} retrieves:
       $$k_\pi = \left\lceil k \cdot        \prod_{j \in [m], \ A_j \in \mathcal{A}_Q} \text{prop}(a_j) \right\rceil,$$
       where $\text{prop}(a_j)$ stands for the proportion of data points having attribute value $a_j$.
        \item Our proposed solutions described in Section~\ref{sec:qpnear}.

   \end{itemize}

   \item \textbf{Combining near point retrieval and fairness.}
   \begin{itemize}[leftmargin=0.5em]
       \item \textsc{SIEVE++}: Use \textsc{SIEVE} for near point retrieval. But uses our bitmap indexing and postprocessing algorithms for fairness.
       \item \textsc{Filter-DiskANN}: Extended by considering multiple protected attributes at the same time. A visited node is considered a valid candidate only if its protected attributes are consistent with the fairness constraints $\hat{\beta}$.
      \item {\tt SAIR} and {\tt JIR} using \textsc{SIEVE}, \textsc{Filter-DiskANN}, or \AlgNN.
       \item Our implemented solutions \AlgFairOne, \AlgFairTwo, \AlgFairThree.
   \end{itemize}
\end{enumerate}

\begin{table}[!htbp]
\centering
\setlength{\tabcolsep}{6pt}
\resizebox{\linewidth}{!}{\begin{tabular}{lccccc}
\toprule
\textbf{Dataset} & \textbf{Modality} & \textbf{\#Points} & \textbf{Dim.} & \textbf{\#Attributes} & \textbf{Distance} \\
\midrule
Audio    & Speech   & $53{,}387$  & 192   & 3 & $\ell_2$\\
FairFace  & Image   & $97{,}698$  & 768   & 3 & $\ell_2$ \\
CelebA$^\ast$   & Image   & $202{,}599$ & 512   & 1--5 & $\ell_2$\\
GloVe    & Text    & $1{,}183{,}514$ & 100   & 3 & cosine\\
Synthetic & --     & $10{,}000+$ & 128   & 3+  & $\ell_2$\\
\bottomrule
\end{tabular}}
\small
\caption{Summary of the datasets. $^\ast$CelebA has up to $40$ protected attributes, we randomly select $5$ protected attributes to split the dataset.}
\vspace{-0.2in}
\label{tab:datasets}
\end{table}

\noindent {\bf Query and parameters.}
For each dataset, we generate $1000$ feasible queries randomly. We vary $k$ (\# query result size), $m$ (\# of protected attributes), distance function ($\ell_2$ and cosine),  hashing related parameters ($\ell$, $\mu$), and graph based parameters (expansion factor (EF), and maximum out-degree $M$). Default values are $m=3$, $k=10$, distance function is $\ell_2$, $\ell=16$,  $\mu=2$, EF=$8$, $M=8$. For the synthetic data experiment, queries are close to the center.

\noindent {\bf Measures.} We present two types of measures.
\begin{enumerate}[leftmargin=2em]
      \item \textbf{Quality.}  We present three quality measures. 
      \begin{itemize}[leftmargin=0.5em]
        \item Successful queries (\%): The fraction of queries that return feasible results.
          \item  Distance approximation factor (DAF): For a query $Q = (q, \hat{\beta})$, let $S_k(Q)$ denote the retrieved top-$k$ set and $S_k^\star(Q)$ be the ground truth distance. The per-query DAF is defined as:
    $$
    \text{DAF}(Q) = \frac{\sum_{x \in S_k(Q)} \norm{x, q}}{\sum_{x \in S^\star_k(Q)} \norm{x, q}}
    $$
    \item Recall@$k$: Measures the overlapping between ground truth and the solution set, divided by $k$:
    $$
    \text{recall@}k= \frac{ \left|S_k(Q) \,\cap\, S^\star_k(Q) \right| }{k}
    $$
      \end{itemize}
     \item \textbf{Cost.} We measure memory consumption of indexing techniques, as well as running time of the preprocessing and query processing algorithms in seconds.
\end{enumerate}



\noindent \paragraph{\bf Summary of Results}
Our experimental results lead to three key observations. First, we demonstrate that graph-based vector search techniques do not satisfactorily retrieve near neighbors when the underlying data is not well clusterable. In such settings, these methods lead to poor distance bounds, leading to degraded recall even in the absence of fairness constraints. Second, we demonstrate that under complex multi-attribute fairness constraints, these existing techniques further break down: they frequently fail to return feasible results. Third, we demonstrate that our proposed framework consistently achieves exact satisfaction of fairness constraints, maintains high recall, and scales effectively across all evaluated settings. Together these results demonstrate that multi-attribute group fairness cannot be reliably achieved by directly adapting existing vector search based methods, and that fairness-aware retrieval requires dedicated indexing and candidate selection mechanisms.



\subsection{Limitations of Graph Based Techniques}\label{exp:lshvsgraph}
Table~\ref{tab:synthetic_topk} reports results averaged over $200$ queries on synthetic data. The results show that HNSW-based methods (e.g., SIEVE++, Filter-DiskANN) consistently fail to retrieve true near-neighbor candidates, leading to poor DAF and recall. This behavior stems from the reliance of graph-based indexing techniques on navigation through regions of high local connectivity. Under a limited search budget, the graph traversal becomes biased toward densely connected regions that may be far from the query, causing the search to get trapped away from the true neighborhood. Consequently, HNSW-based methods return candidates that are densely connected but distant from the query, resulting in low recall and large DAF. In contrast, LSH-based methods are agnostic to local graph connectivity and instead rely on randomized projections for candidate generation. This enables them to surface sparse yet relevant near neighbors even under severe partition imbalance, yielding substantially better recall and distance guarantees.

\begin{table}[!htbp]
\centering
\begin{threeparttable}
\resizebox{\linewidth}{!}{\begin{tabular}{lccc}
\toprule
\textbf{Method} 
& \textbf{$k$-th DAF} $\downarrow$ 
& \textbf{Recall@10} $\uparrow$ 
& \textbf{Recall@1} $\uparrow$ \\
\midrule
\textbf{Ours} 
& \textbf{1.01} 
& \textbf{0.93} 
& \textbf{0.98} \\
HNSW-based methods 
& 8.35 
& 0.63 
& 0.63 \\
\bottomrule
\end{tabular}}
\small
\caption{\small Comparison between graph based (\textsc{SIEVE++}, \textsc{Filter-DiskANN}) and LSH based indexing techniques in  top-$k$ nearest neighbor retrieval with $k=10$.}
\vspace{-0.2in}
\label{tab:synthetic_topk}
\end{threeparttable}
\end{table}

\vspace{-0.2in}
\subsection{Query Processing}\label{exp:queryp}
We use the real datasets to evaluate the quality of query results  in terms of distance and fairness constraints. These results are reported as the average over $1000$ queries.

\subsubsection{Qualitative Evaluation on Distance}\label{qp:q1}
Figure~\ref{fig:qualitative-approx} shows that all algorithms perform well with varying $k$. In particular, at $k=20$, \textsc{Filter-DiskANN} reaches an average DAF of 1.07, compared with 1.09 for \textsc{SIEVE++} and 1.11 for \AlgNN, while brute force remains exactly 1.00 by definition. A similar ordering holds at $k=15$.

However, varying the number of protected attributes~($m$) reveals that both \textsc{SIEVE++} and \textsc{Filter-DiskANN} degrade sharply for $m=4$ and $m=5$. In these settings, DAF drops to zero, reflecting frequent infeasibility and failed query answering. As 
$m$ increases, required attribute combinations fall into sparse or weakly connected regions for these algorithms, preventing the retrieval of sufficient candidates to satisfy fairness constraints and causing infeasibility. This behavior highlights a fundamental limitation of these methods, precisely the setting that \AlgNN is designed to support.

\begin{figure}[!htbp]
    \centering
    \includegraphics[width=1.0\linewidth, trim=0.7cm 0.09cm 1.6cm 0.2cm, clip]{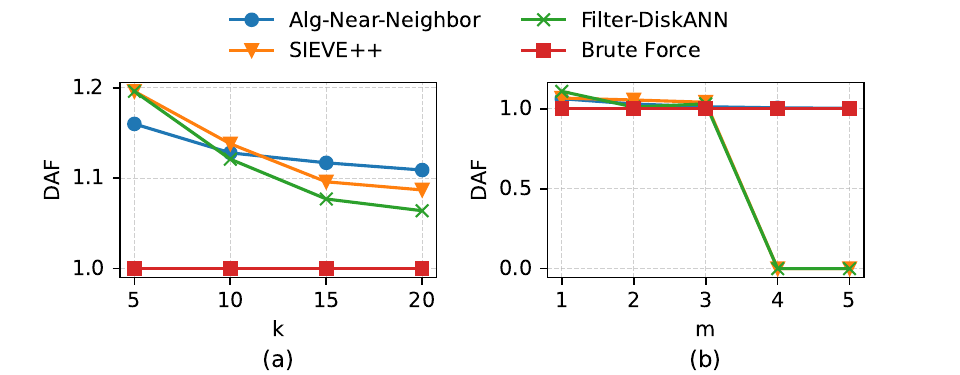}
    \caption{\small Average DAF varying (a) $k$ on Audio dataset, and (b) $m$ on CelebA dataset. Zero DAF means infeasible solutions.}
    \label{fig:qualitative-approx}
\end{figure}





\subsubsection{Qualitative Evaluation on Fairness.} \label{qp:q2}
In these experiments, we evaluate the quality of the proposed disk based solutions with \AlgNN, as well as fairness baseline and retrieval techniques, like {\tt SAIR}, and {\tt JIR}. Figures~\ref{fig:qualitative-fairness} (a) and (b) present recall varying $k$ and $m$, respectively, whereas, (c) and (d) present those for \% successful queries. 
Figures~\ref{fig:qualitative-fairness} (a) and (c) shows that \AlgNN consistently achieves recall@$k=1.0$ and return $100\%$ successful queries, 
Figure~\ref{fig:qualitative-fairness} (a) shows that \AlgNN and \textsc{SIEVE++} perform better at larger $k$ than  {\tt SAIR}, {\tt JIR}, and \textsc{Filter-DiskANN} baselines, indicating that these latter techniques can not handle multi attribute fairness constraints when independence does not satisfy. 
As expected, Figures~\ref{fig:qualitative-fairness} (b) and (d) demonstrate that quality degrades in general when $m$ increases, however, \AlgNN outperforms all other algorithms.


\begin{figure}[ht]
    \centering
    \includegraphics[width=1.0\linewidth, trim=0.69cm 0.3cm 1.8cm 0.25cm, clip]{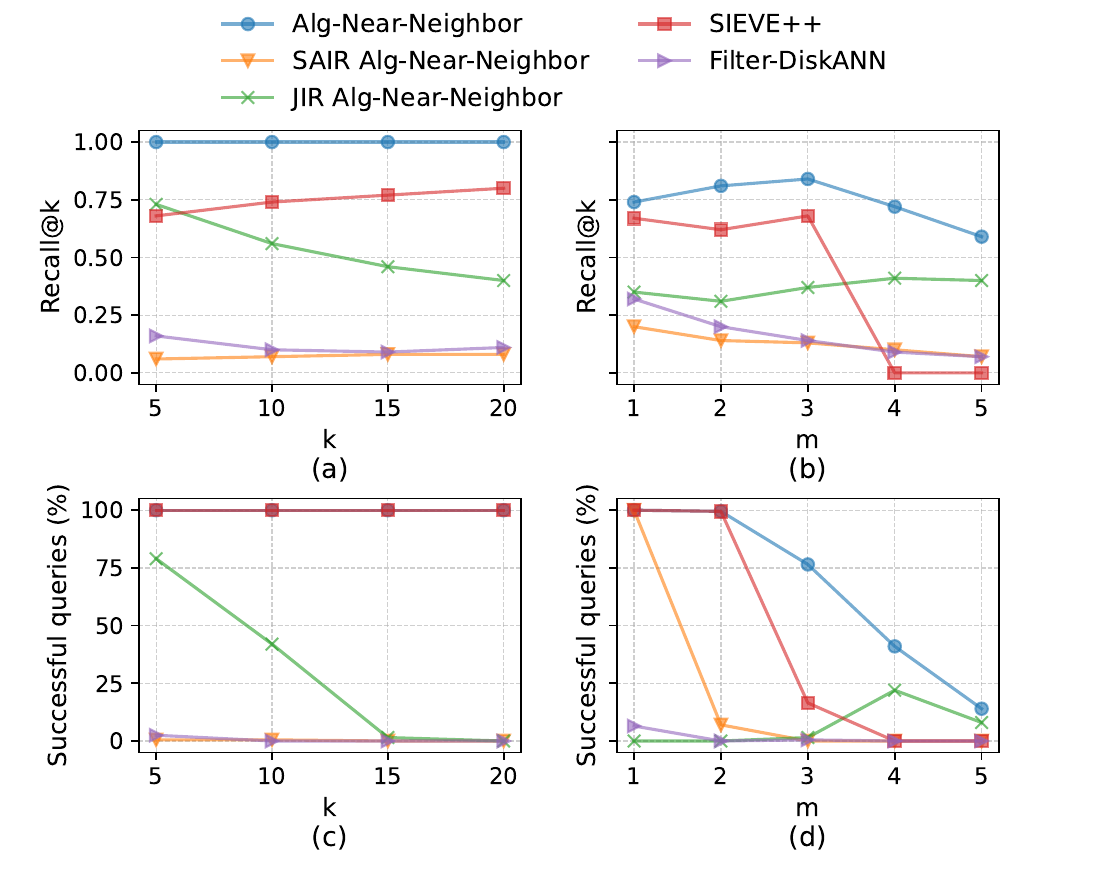}
    \caption{\small Comparing different fairness based baselines. (a) recall@$k$ with varying $k$, (b) recall@$k$ with varying $m$, (c) successful queries (\%) with varying $k$, (d) successful queries (\%) with varying $m$. We vary $k$ on FairFace, $m$ on CelebA datasets, and set $\ell=128$ for the LSH-based methods.}
    \label{fig:qualitative-fairness}
\end{figure}







\vspace{-0.1in}
\subsubsection{Qualitative Evaluation on Fairness for LSH Based Techniques}\label{qp:q3}

Figure~\ref{fig:fairness-baselines} evaluates how LSH parameters (hash bucket width $w$ and the number of hash tables $\ell$) affect fairness-constrained retrieval.
Across all configurations, \AlgNN consistently achieves 100\% successful queries, confirming that its Cartesian-attribute indexing always retrieves enough candidates to satisfy the fairness constraints.
In contrast, {\tt SAIR} and {\tt JIR} frequently fail to find feasible solutions: {\tt SAIR} succeeds in roughly 40\% of queries, and {\tt JIR} in only 20\%, indicating that both methods struggle to cover all required partitions as parameters change.

A similar trend appears in recall@$k$.
\AlgNN maintains the highest recall and benefits from increasing $\ell$, while recall of {\tt JIR} remains lower  and {\tt SAIR} reaches to $0$, due to its independent attribute indexing. Furthermore, {\tt JIR} and {\tt SAIR} show minimal sensitivity to either $w$ or $\ell$, indicating that their designs do not effectively leverage additional hashing resources. Overall, \AlgNN preserves multi-attribute structure and improves with hashing capacity, compared to the baselines.

\begin{figure}[h]
    \centering
    \includegraphics[width=1.0\linewidth, trim=0.75cm 0.25cm 1.8cm 0.27cm, clip]{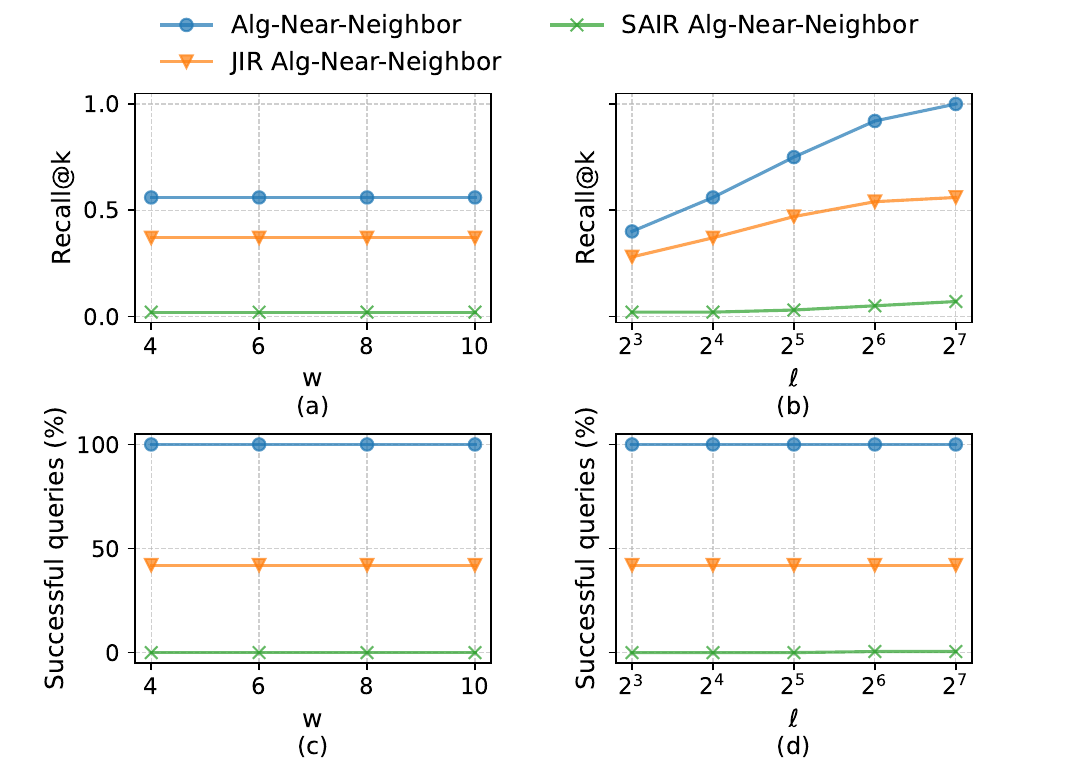}
    \caption{\small Qualitative evaluation varying LSH parameters using FairFace dataset. (a) recall@$k$ varying $w$, (b) recall@$k$ varying $\ell$, (c) successful queries (\%) varying $w$, (d) successful queries (\%) varying $\ell$.}
    \label{fig:fairness-baselines}
\end{figure}








\subsubsection{Scalability Evaluation}\label{qp:scale}
Sections~\ref{qp:q1},~\ref{qp:q2}, and~\ref{qp:q3} have already highlighted the limitations of existing baselines in producing high-quality results for the \MAtFair problem. 
In this section, we therefore present scalability of our solutions only with brute-force.

Figure~\ref{fig:post-time} reports the latency required to find the set of $k$ satisfying the fairness constraints.
In Figure~\ref{fig:post-time} (a), the query processing time---including candidate retrieval and post-processing---of \AlgNN remains nearly flat as $k$ increases, whereas the brute-force solver grows noticeably and remains $3$--$4\times$ slower, reflecting the efficiency of our indexed retrieval.
In Figure~\ref{fig:post-time} (b), we fix $k=20$ and $m=3$ and compare query time across datasets.
\AlgNN consistently maintains superior latency on all datasets, while brute force becomes prohibitively slow on larger data such as GloVe, demonstrating the scalability gap between the two approaches.

Figure~\ref{fig:post-time} (c) shows a small increase in post-processing time from $m=1$ to $m=2$, followed by a significant jump at $m=3$.
This reflects the algorithmic behavior of the selection problem: with $m \le 2$, the feasibility check can be solved efficiently in polynomial time, whereas $m \ge 3$ requires ILP-based solver for the NP-hard formulation, leading to substantially higher latency.
In Figure~\ref{fig:post-time} (b), we compare \AlgPPTwo and \AlgPPThree .
Across all values of $k$, \AlgPPTwo, which uses network flow maintains roughly $4\times$ lower latency than \AlgPPThree, which uses ILP, highlighting its practical advantage as a scalable post-processing strategy .





\begin{figure}[!htbp]
    \centering
    \includegraphics[width=1.0\linewidth, trim=0.27cm 0.30cm 0.25cm 0.21cm, clip]{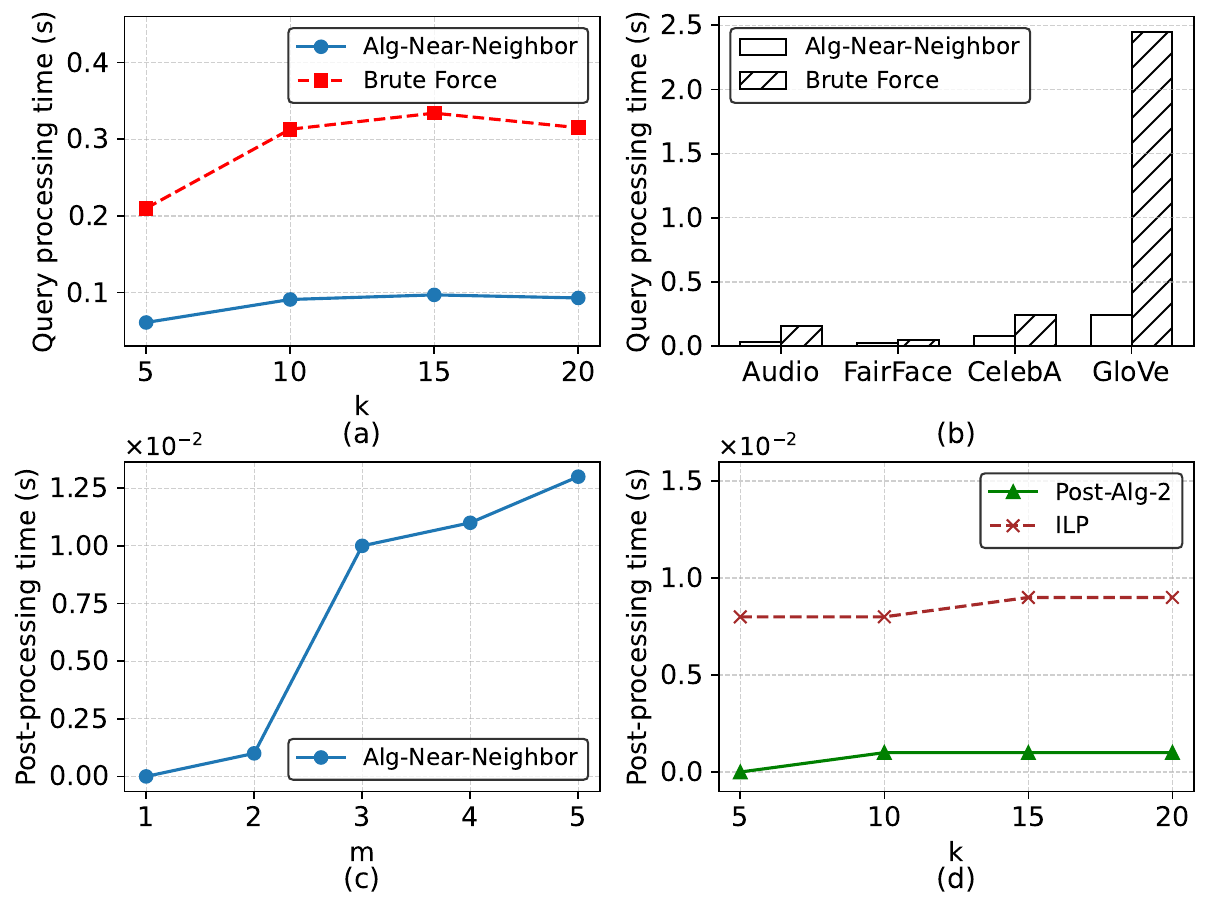}
    \caption{\small  (a) Query processing time (retrieval + post-processing) as $k$ varies (with $m=2$). (b) Comparing query processing time across datasets (with fixed $k=20$ and $m=3$). (c) \AlgNN overhead when increasing $m$ for $k=20$, (d) Comparison between Network Flow and ILP at varying $k$ for $m=2$.}
    \label{fig:post-time}
\end{figure}

Figure~\ref{fig:candidates} compares the fraction of scanned candidates between \AlgNN and brute-force varying $m$, $k$.
Across all settings, \AlgNN scans dramatically fewer candidates, often under 10\% of the dataset, while brute force must examine a substantially larger portion, exceeding 60\% for small $m$.
As $m$ increases, bitmap-based filtering in \AlgNN becomes more selective, further shrinking the candidate set, whereas brute force continues to scale poorly.

\begin{figure}[!htbp]
    \centering
    \includegraphics[width=1.0\linewidth, trim=0.27cm 0.32cm 0.25cm 0.21cm, clip]{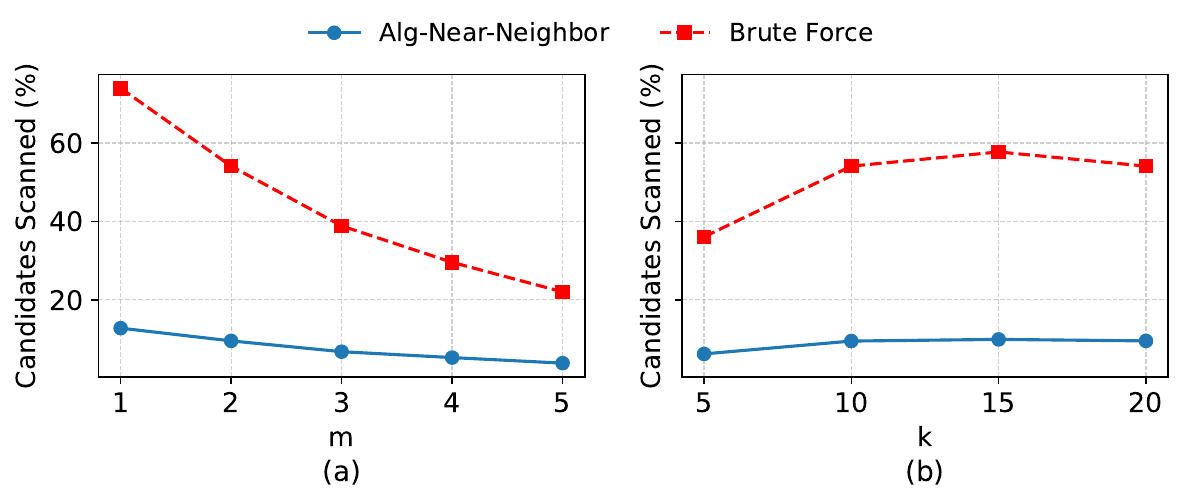}
    \caption{\small Comparison between \AlgNN and brute-force on fraction of candidates scanned to satisfy fairness constraints at varying (a) $m$ and (b) $k$.}
    \label{fig:candidates}
\end{figure}






\subsection{Evaluation of Preprocessing}\label{exp:prep}
In this section, we evaluate the preprocessing costs required to support fairness-aware retrieval.
Our focus is specifically on the overhead introduced to ensure fairness.
For this reason, we do not include the graph-based methods in this comparison.

\subsubsection{Comparison with Fairness Based Baselines}

Figure~\ref{fig:preprocessing-cost} compares the preprocessing overhead of different LSH-based indexing strategies, reporting both index construction time and storage cost across three datasets.
Overall, \AlgNN incurs lower indexing time and storage overhead, as it avoids duplicating hashed data points across multiple per-attribute indexes, which is inherent in \AlgNN ({\tt SAIR}).
An exception arises for the FairFace dataset, where \AlgNN requires higher storage due to the large number of Cartesian-attribute partitions, whose cumulative indexing cost---including stored hash functions---outweighs the storage savings from eliminating data duplication.

\begin{figure}[t]
    \centering
    \includegraphics[width=1.0\linewidth, trim=1.42cm 1.6cm 2.19cm 0.42cm, clip]{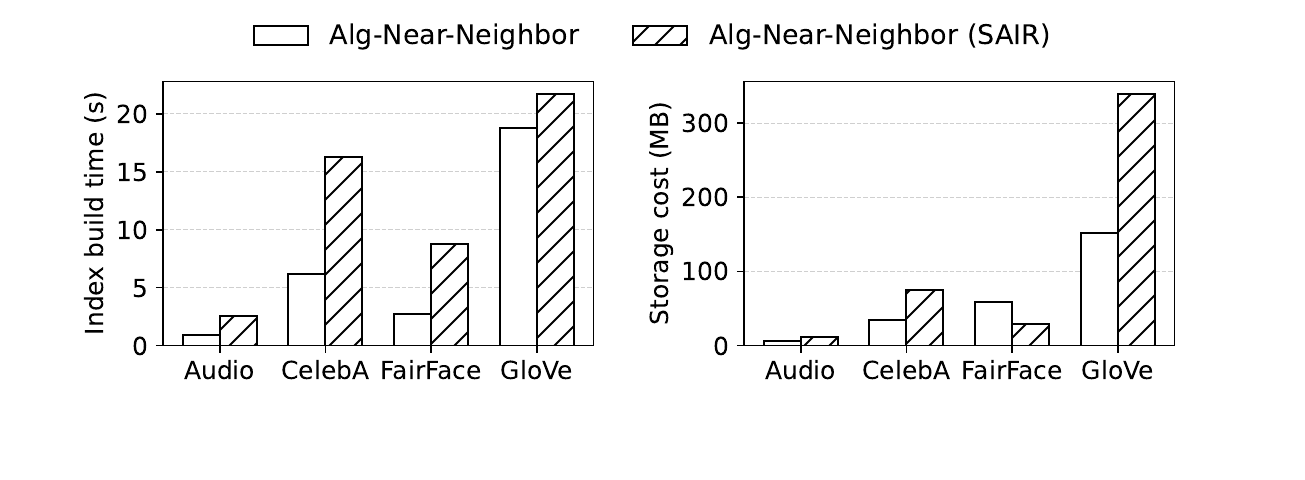}
    \caption{Preprocessing cost of fairness-aware LSH indexing.}
    \label{fig:preprocessing-cost}
\end{figure}

\subsubsection{Fine-grained Evaluation Varying LSH parameters}
Figure~\ref{fig:preprocessing-vary-LSH} presents the analysis of preprocessing cost as key LSH parameters vary, specifically the number of hash tables $\ell$ and the concatenation length $\mu$, across the real-world datasets and report both index build time and storage cost. Only the GloVe dataset uses cosine distance, whereas, the remaining ones use $\ell_2$ distance. Increasing $\ell$ takes higher build time and storage cost as demonstrated in Figures~\ref{fig:preprocessing-vary-LSH} (a) and (c), as expected. GloVe turns out to be most expensive because of its underlying distance function. Similarly, increasing the concatenation length $\mu$ results in a higher pre-processing cost as shown in Figures~\ref{fig:preprocessing-vary-LSH} (b) and (d), although the growth rate is comparatively milder than that of $\ell$.
An exception occurs for the storage cost for GloVe dataset, where the curve shows a minimal increase in storage requirement as $\mu$ increase. Overall, these results demonstrate the efficacy of the pre-processing technique.

\begin{figure}[t]
    \centering
    \includegraphics[width=1.0\linewidth, trim=0.75cm 0.16cm 1.9cm 0.54cm, clip]{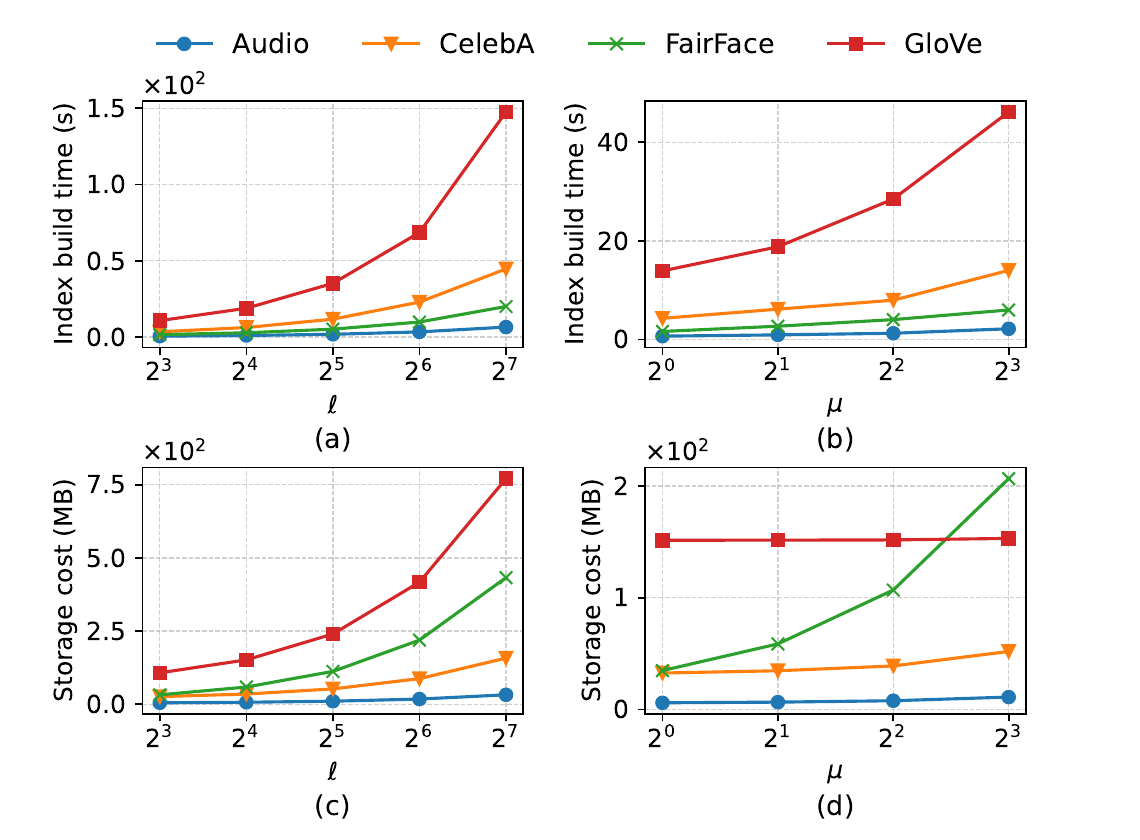}
    \caption{\small Preprocessing cost of \AlgNN on various datasets with varying $\ell$ and $\mu$.}
    \label{fig:preprocessing-vary-LSH}
\end{figure}

\section{Related Work}~\label{sec:relatedwork}
\noindent {\bf NN Search in Vector Databases}
Vector databases have become essential infrastructure for AI applications, enabling efficient storage and retrieval of text, image, and multimodal embeddings~\cite{han2023comprehensive, jing2024large, pan2024survey}. Systems such as {\tt Milvus}~\cite{wang2021milvus}, {\tt Qdrant}~\cite{qdrant2025}, {\tt SingleStore\text{-}V}~\cite{chen2024singlestore}, and {\tt TigerVector}~\cite{liu2025tigervector} pair vector indexes with database functionality to support hybrid similarity{+}filter queries, playing a central role in RAG and recommendation pipelines where vector search quality directly impacts downstream performance.

Current indexing methods fall into three families~\cite{pan2024survey}: LSH-based~\cite{huang2015query,indyk1998approximate,datar2004locality,gan2012locality}, graph-based~\cite{malkov2018efficient,simhadridiskann,patel2024acorn,li2025sieve,acorn1,acorn2,acorn3}, and tree-based~\cite{dasgupta2008random,dasgupta2013randomized,marius2009flann}. While early work offered no query predicates and later systems introduced simple filters~\cite{li2025sieve,acorn1,acorn2,acorn3,patel2024acorn}, none of these methods as is support the complex count-based constraints required for group fairness. LSH approaches provide formal guarantees, whereas graph methods are heuristic but effective for clusterable data; however, their designs do not naturally generalize to fairness-aware retrieval.

{\em We show that these indexing strategies cannot be used as-is to enforce complex group fairness constraints, whereas our postprocessing framework could  be paired with any of them.}

\noindent{\bf Group Fairness}
Group fairness in search, recommendation, and content curation aims to ensure proportional representation of protected groups in retrieved results~\cite{bei2022candidate,islam2022satisfying,oesterling2024multigroup}. Existing methods typically apply fairness through post-processing or re-ranking of a candidate set produced by standard ANN search. Recent work proposes Multi-Group Proportional Representation (MPR)~\cite{oesterling2024multigroup}  as a metric, and optimization goal, for proportional fairness over intersectional groups. In subsequent work, Jung et al.~\cite{jung2025multi} adapt MPR to text-to-image generation, using MPR to measure and train for representational balance across intersectional groups in model outputs.

Multi-attribute fairness further increases complexity\cite{celis2019classification,islam2022satisfying}: \cite{islam2022satisfying} introduces a re-ranking framework that enforces proportional representation constraints in top-$k$ results by selectively substituting items in a baseline ranking. The ballot substitution problem is shown to be weakly NP-hard for 2-attributes~\cite{islam2022satisfying}, and three attributes reduce to the strongly NP-hard 3-dimensional matching problem. This motivates hybrid approaches that combine approximate retrieval with exact constrained selection.

{\em MPR is a powerful metric and learning/selection objective for multi-group proportionality at the distributional level, but does not necessarily certify that per-attribute counts are exactly met for every query, as studied in our problem. }

{\bf Mitigating bias in Embedding space.}
An orthogonal approach to our effort are techniques that are used to produce vector embeddings of text and images that target diversity or mitigate bias. These include modifying loss
functions to encourage group fairness~\cite{debias1}, adversarial training~\cite{debias2}, or disentangling representations~\cite{debias3}. Many popular embedding models are not trained with such approaches, so post-processing methods for fairness have also been
developed for multimodal models such as CLIP ~\cite{clip}, including CLIP-clip~\cite{clip1}, CLIP-debias~\cite{clip2},
and FairCLIP~\cite{wang2022fairclip}. However, these are limited to vision model and are unable to satisfy any user defined fairness constraints.
\vspace{-0.1in}
\section{Conclusion}
This work introduces a principled framework for efficient $k$-NN search in vector databases under multi-attribute group-fairness constraints. We formalize the \MAtFair problem and show that it is computationally hard for three or more protected attributes, motivating hybrid retrieval and optimization techniques. Our indexing design couples fairness-aware bitmap partitioning with LSH-based candidate generation, ensuring efficiency while maintaining fairness feasibility. We also develop a family of query-processing algorithms with provable guarantees based on combinatorial optimization. Experiments on large-scale vector datasets show that existing methods cannot be directly extended to address this problem, whereas our framework augments them effectively to support this new, fairness-aware retrieval task.

\bibliographystyle{ACM-Reference-Format}
\bibliography{references}

\end{document}